\documentclass[
 superscriptaddress,
 amsfonts,amsmath,amssymb,
 aps,prx,twocolumn,
 floatfix
]{revtex4-2}

\usepackage{graphicx}
\usepackage{dcolumn}
\usepackage{physics}
\usepackage{bm}
\usepackage{dsfont,xcolor}
\usepackage[colorlinks=true,linkcolor=blue,citecolor=blue]{hyperref}
\usepackage{amsthm}

\newtheorem{lemma}{Lemma}

\begin{document}

\title{Dilute measurement-induced cooling into many-body ground states}

\author{Josias Langbehn}
\affiliation{
Dahlem Center for Complex Quantum Systems and Fachbereich Physik, Freie Universit\"at Berlin, Arnimallee 14, 14195 Berlin, Germany 
}
\author{Kyrylo Snizhko}
\affiliation{Univ. Grenoble Alpes, CEA, Grenoble INP, IRIG, PHELIQS, 38000 Grenoble, France}
\author{Igor Gornyi}
\affiliation{Institute for Quantum Materials and Technologies and Institut für Theorie der Kondensierten Materie, Karlsruhe Institute of Technology, Karlsruhe 76131, Germany}
\author{Giovanna Morigi}
\affiliation{Theoretical Physics, Department of Physics, Saarland University, 66123 Saarbrücken, Germany}
\author{Yuval Gefen}
\affiliation{Department of Condensed Matter Physics, Weizmann Institute of Science, Rehovot 7610001, Israel}
\author{Christiane P. Koch}
\email{christiane.koch@fu-berlin.de}
\affiliation{
Dahlem Center for Complex Quantum Systems and Fachbereich Physik, Freie Universit\"at Berlin, Arnimallee 14, 14195 Berlin, Germany 
}

\date{\today}
\begin{abstract}
Cooling a quantum system to its ground state is important for the characterization of non-trivial interacting systems, and in the context of a variety of quantum information platforms. In principle, this can be achieved by employing  measurement-based passive steering protocols, where the steering steps are predetermined and are not based on measurement readouts. However, measurements, i.e., coupling the system to auxiliary quantum degrees of freedom, is rather costly, and protocols in which the number of measurements scales with system size will have limited practical applicability. Here, we identify conditions under which measurement-based cooling protocols can be taken to the dilute limit. For two examples of frustration-free one-dimensional spin chains, we show that steering on a single link is sufficient to cool these systems into their unique ground states. We corroborate our analytical arguments with finite-size numerical simulations and discuss further applications.
\end{abstract}

\maketitle

\section{Introduction}

The ability to control many-body quantum systems is a prerequisite for implementing quantum information processing, but also, more generally, for advances across quantum physics from atomic and molecular all the way to condensed matter physics. Particular challenges are to identify control strategies that are both scalable with increasing system size and robust to parameter fluctuations. The challenge of robustness can be addressed by so-called quantum reservoir engineering~\cite{PoyatosPRL96} where coupling to a reservoir induces desired dissipation. In the long-time limit, the system can then be driven into a predesigned target state, irrespective of the initial state. For many-body quantum systems, this allows for preparing pure and possibly highly entangled states or driving the system into non-trivial quantum phases~\cite{DiehlNatPhys08, kraus_preparation_2008, verstraete_quantum_2009, DiehlNatPhys2011, mi2023,Brown_etal-Kamal:2022}.

When seeking to exploit engineered dissipation as a resource for, e.g., quantum  computation~\cite{kraus_preparation_2008, verstraete_quantum_2009}, one is often faced with the problem that the required couplings with the reservoir are difficult or even impossible to design. In other words, natural dissipation processes may simply not allow for the system-reservoir interactions needed for a given target state~\cite{kraus_preparation_2008}. An alternative is to leverage quantum measurements to induce the desired dissipative dynamics~\cite{HarringtonNatRevPhys2022}. In quantum measurements, the quantum system of interest is coupled to a ``meter'' (``detector'') for a time during which the system and meter become entangled (at least weakly), at which point the interaction is switched off~\cite{wisemanBook}. A textbook example is given by a quantized electromagnetic field mode that interacts with Rydberg atoms as ``meters''~\cite{HarocheBook}. Depending on the initial state of the atoms, the quantum system can be cooled to zero temperature~\cite{HarocheBook} or prepared in an entangled state~\cite{Pielawa:2007}. The concept of engineered environment has found a wide range of applications, for instance, for cooling a spin chain to the ground state by using a qubit as a meter whose frequency is tuned to the lowest energy gap~\cite{raghunandan2020}. It is also central to protocols based on homogenization of quantum systems \cite{BurgarthPRL2007, BurgarthPRA2007} and to ``algorithmic'' or digital cooling in quantum simulators~\cite{Boykin2002, Metcalf2020, Zaletel2021}.

For a single quantum degree of freedom interacting with one or more meters, such as the electromagnetic field mode interacting with a beam of atoms~\cite{HarocheBook}, any desired state can be prepared with suitably optimized classical drives~\cite{RojanPRA2014}. 
As another example, combining measurement-based dynamics with unitary evolution, it is possible to steer a spin-1/2 system to any desired state on or inside the Bloch sphere~\cite{roy_measurement-induced_2020, KumarPRR2020}.
Such a high level of control cannot be expected for many-body quantum systems where the multipartite nature introduces competing timescales and restrictions on the reachable states. For example, only certain classes of states will be attainable with quasi-local couplings~\cite{Perez-GarciaQIC2008, JohnsonQIP2016} and systems with frustrated local Hamiltonians are not amenable to passive steering~\cite{ticozzi_steady-state_2014}. 

Quantum measurements can nevertheless be exploited in multiple ways, which include, among many others, inducing phase transitions~\cite{FisherAnnuRev2023, Noel2022a, Koh2022, hoke2023}, realizing cooling with quantum hardware~\cite{PollaPRA2021, FengPRA2022, matthies2023, kishony2023, maurya2023}, protect quantum states and dynamics~\cite{Biella2021,Albash_Lidar18,Menu_etal22}, modifying the entanglement structure~\cite{paviglianiti2023}, or steering a quantum system from an arbitrary initial state toward a chosen target state~\cite{roy_measurement-induced_2020, KumarPRR2020, HerasymenkoPRXQ2023}. 

The Affleck-Lieb-Kennedy-Tasaki (AKLT) model~\cite{AKLT1,AKLT2} has served as a major paradigm for both quantum reservoir engineering~\cite{kraus_preparation_2008} and measurement-induced dynamics, including state engineering~\cite{roy_measurement-induced_2020,MurtaPRR2023, SmithPRXQ2023, chen2023highfidelity,chen2023efficient, puente2023} in many-body quantum systems; 
similar protocols have also been brought forward for the two-dimensional Kitaev model, a spin liquid~\cite{SriramPRB2023, LavasaniPRB2023}.
In this respect, the AKLT model and its relatives could be considered as very instrumental tools that bridge the field of measurement-induced state preparation with that of strongly correlated systems.

As much as versatility and robustness make measurement-based control appealing, for many-body quantum systems, this comes at a cost. The first distinction is based on whether or not the measurement outcomes are used for control~\cite{wisemanBook}. In the context of steering, this is referred to as passive~\cite{roy_measurement-induced_2020} vs active~\cite{HerasymenkoPRXQ2023} steering, with passive steering including measurement-induced cooling~\cite{roy_measurement-induced_2020, PollaPRA2021, puente2023}. While active steering may leverage decision-making policy~\cite{HerasymenkoPRXQ2023, morales2023}
or optimization~\cite{sayrin2011} to facilitate the approach of the target state, the actual implementation of the feedback requires additional resources, on top of the auxiliary quantum systems serving as meters. For many-body systems, this will quickly become challenging, if not unfeasible. In fact, already for passive steering protocols, one may wonder how far these can be pushed given their requirement of the number of meters that is extensive in system size, cf. Fig.~\ref{fig:schematic_cooling}(a).

\begin{figure}[tbp]
\includegraphics[width=1\columnwidth]{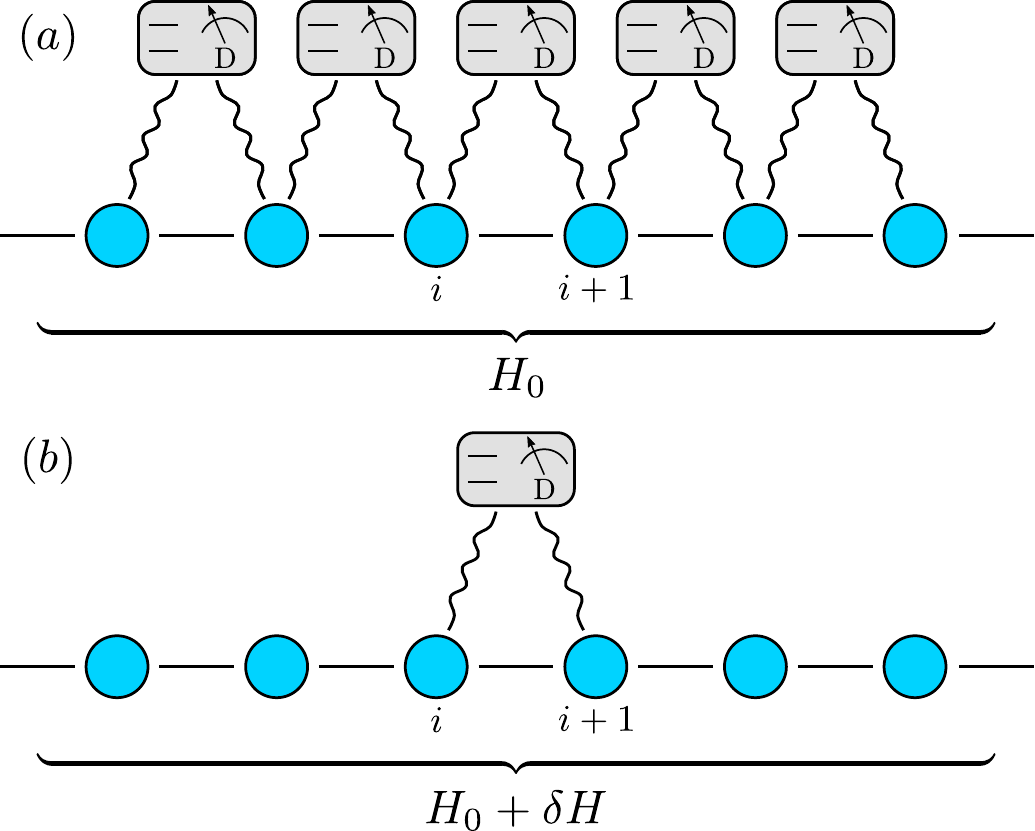}%
\caption{Measurement-induced cooling of a many-body system described by Hamiltonian $H_0$ (shown as one-dimensional for simplicity). (a) Setup as discussed in Ref.~\cite{roy_measurement-induced_2020} with an extensive number of detectors $D$, each one coupled to a pair of sites. (b) Dilute cooling in the extreme limit of a single detector $D$: The interplay between the local detector-system coupling and possibly additional coherent interactions $\delta H$ allows for driving the system into the ground state of $H_0$.}  \label{fig:schematic_cooling}
\end{figure}

Here, we adopt the perspective that detectors are a precious resource whereas the ability to engineer (quasi-local) interactions within a many-body system can be taken as given, since this is a prerequisite to implement the many-body model itself. The detectors are key to realizing measurement-induced dynamics, but how many are actually needed? For the paradigmatic AKLT model, a chain of spin-1 particles, we show that measuring on a single pair of two neighboring sites is sufficient to cool the entire chain into its ground state, cf. Fig.~\ref{fig:schematic_cooling}(b). The reason is that, for sufficiently large chains, the nearest-neighbor couplings allow for propagating any excitation to the measured link where the excitation is dissipated. While the structure of the AKLT Hamiltonian is such that it ensures the propagation for all conceivable excitations, in general, this is, of course, not the case. We, therefore, identify which quasi-local Hamiltonians should be added to engineer the necessary population flow and exemplify our theoretical framework for dilute cooling of a many-body system with a second example, the Majumdar-Ghosh model~\cite{MajumdarGhoshOriginal}. Finally, we discuss the hierarchy of timescales for the excitation propagation and localized measurements, in order to identify the scaling of the cooling time with the system size, as well as time dependence of cooling in the limit of infinite systems.

The remainder of the paper is organized as follows. Section~\ref{sec:measuremen_induced_cooling} presents the overall concept of measurement-induced cooling, starting with a brief recap of passive steeering~\cite{roy_measurement-induced_2020} and its description in terms of a master equation in Sec.~\ref{subsec:recap}. In Sec.~\ref{subsec:Lchoice}, we introduce the notions of local hot and cold subspaces on the link(s) that shall be cooled. This gives rise to conditions on the system-detector couplings to be permissible for a given target state. Another key ingredient for many-body cooling is the cooling rate and we explain how to estimate it as a function of the system size in Sec.~\ref{subsec:rate}. 
In Sec.~\eqref{sec:realization} we analyze numerically two distinct models defined on finite chains, and demonstrate that cooling a single link is sufficient to drive into the respective many-body ground state.For the AKLT model, this is possible for sufficiently large chains, as we discuss in Sec.~\ref{subsec:aklt}, while we show in  Sec.~\ref{subsec:majumdar_ghosh}, using the Majumdar-Ghosh model, that dilute cooling can also be used to selectively drive into one state within a  degenerate ground state manifold for sufficiently large chains. Dilute cooling is not limited to these two examples and we formulate a general framework in Sec.~\ref{sec:cooling_conditions}.
In Sec.~\ref{subsec:Hchoice} we formalize the condition to preserve the target for all conceivable interactions within the system. This, together with the corresponding condition on the system-detector couplings, allows us to state necessary and sufficient conditions for dilute cooling in  Sec.~\ref{subsec:necessarycond} and \ref{subsec:suffcond}, respectively. 
With those conditions at hand, we show how to prepare the AKLT state in small chains in Sec.~\ref{subsec:smallAKLT}. We summarize the general framework for dilute measurement-induced cooling in the form of a "recipe" in Sec.~\ref{sec:disc} before concluding in Sec.~\ref{sec:concl}.

\section{Measurement-induced cooling \label{sec:measuremen_induced_cooling}}

We first briefly review the protocol of Ref.~\cite{roy_measurement-induced_2020} for steering the state of a quantum system from an arbitrary initial state toward a chosen target state by coupling it to an extensive number of auxiliary quantum degrees of freedom ("detectors") in Sec.~\ref{subsec:recap}. Then we lay the foundations for dilute cooling by introducing the notions of cold and hot subspaces on the cooled sites in Sec.~\ref{subsec:Lchoice}. When reducing the number of detectors, a key question is how fast the target state is approached and we discuss  in Sec.~\ref{subsec:rate} how to estimate the cooling rate (or time) to reach the target state. 

\subsection{Recap of passive steering}\label{subsec:recap}

For simplicity, the auxiliary quantum degrees of freedom are taken to be qubits. 
The total Hamiltonian for the composition of system and detector is written as 
\begin{equation}
    \label{eq:Htotal}
H=H_{\text{s}}\otimes \mathbb{I}_d + H_\text{s-d} + \mathbb{I}_s\otimes H_{\text{d}}
\end{equation}
where $H_s$, $H_{s-d}$, and $H_d$ act respectively on the system only, on the system and detector, and on the detector only and $\mathbb{I}_s$ ($\mathbb{I}_d$) is the identity operator on the system (detector) Hilbert space. In the following, we drop the explicit reference to the identity operators and work in the interaction picture with respect to $H_\text{s} + H_\text{d}$. The protocol consists of repeatedly performing the following steps~\cite{roy_measurement-induced_2020}:
\begin{enumerate}
\item Each detector qubit is prepared in a fixed pure state, independent of the system state. Detector and system are in the separable state $\rho_{\text{d}}\otimes\rho_{\text{s}}(t)$, with $\rho_{\text{d/s}}$ the state of the detector qubits, respectively system, alone.
\item System and detector qubits are coupled during a time interval $\delta t$. The time evolution generated by the corresponding interaction $H_{\text{s-d}}$ is given by
  \begin{equation}
    \rho_{\text{s-d}}\left(t+\delta t\right)= e^{-iH_{\text{s-d}}\delta t}\rho_{\text{d}}\otimes\rho_{\text{s}}\left(t\right)e^{iH_{\text{s-d}}\delta t}\label{eq:infenitesimal_increment}\,, 
  \end{equation}
  where $\rho_{\text{s-d}}$ denotes the joint state of the system and detector qubits. 
\item The interaction is switched off and the detector qubits are discarded.
  This is equivalent to a projective measurement of the detector qubits, with an unbiased average over all measurement outcomes. The system state is then obtained by tracing out the detector qubits,
  \begin{equation}
\rho_{s}\left(t+\delta t\right) =\text{Tr}_{d}\rho_{\text{s-d}}\left(t+\delta t\right)    
  \end{equation}
\end{enumerate}

The dynamics of the system evolving under this protocol is obtained as follows~\cite{roy_measurement-induced_2020}. For simplicity, we discuss the derivation of the equation of motion for a single detector qubit; the extension to several detectors is straightforward.
The choice of the initial state of the detector qubit, $\ket{\phi_0}$, induces a partition on the detector Hilbert space $\mathcal{H}_{d}$ into two orthogonal subspaces, $\mathcal{H}_{\text{d}}=\mathcal{D}_{0}\oplus\mathcal{D}_{1}$, with $\mathcal{D}_{0}$ spanned by $\ket{\phi_{0}}$ and $\mathcal{D}_{1}=\mathcal{D}_{0}^{\perp}$ spanned by $\ket{\phi_1}$ such that $\braket{\phi_0}{\phi_1}=0$. 
Likewise, the composite system-detector Hilbert space $\mathcal{H}=\mathcal{H}_{s}\otimes\mathcal{H}_{\text{d}}$ is partitioned into subspaces
\begin{equation}
    \mathcal{H}_{i}=\mathcal{H}_{\text{s}}\otimes\mathcal{D}_{i}
\end{equation}
for $i=0,1$ such that
\begin{equation}
    \mathcal{H}=\mathcal{H}_{0}\oplus\mathcal{H}_{1}\,.
\end{equation}
The system-detector interaction can be represented as~\footnote{In general, this may also include diagonal terms, leading to a Lamb shift in the system Hamiltonian $H_{\text{s}}$.}
\begin{align}
H_{\text{s-d}}= & \begin{pmatrix}0 & \sqrt{\gamma} \tilde{L}^{\dagger}\\
\sqrt{\gamma} \tilde{L} & 0
\end{pmatrix}\,,
\end{align}
where $\gamma$ parameterizes the coupling strength between system and detector and is also referred to as the \emph{dissipation rate}. The non-Hermitian operators $\tilde{L}, \tilde{L}^\dagger$ act on the system whenever the detector state changes from $\ket{\phi_0}$ to $\ket{\phi_1}$ or vice versa.
For sufficiently short intervals $\delta t$, Eq.~\eqref{eq:infenitesimal_increment} can be expanded to second order in $\delta t$. Tracing out the detector qubit results in 
\begin{equation}
\frac{\rho_{\text{s}}\left(t+\delta t\right)-\rho_{\text{s}}\left(t\right)}{\delta t}= \gamma \delta t \left(\tilde{L}\rho_{\text{s}}\left(t\right)\tilde{L}^{\dagger}-\frac{1}{2}\left\{ \tilde{L}^{\dagger}\tilde{L},\rho_{\text{s}}\left(t\right)\right\} \right)    \,.
\end{equation}
Taking the limit of continuous measurements ($\delta t\to 0$) while keeping $L=\tilde{L}\sqrt{\delta t}$ constant yields an equation of motion for the reduced state of the system which is of Gorini-Kossakowski-Sudarshan-Lindblad (GKSL) form \cite{Breuer_Petruccione07,Rivas_Huelga}. The generalization to multiple detector qubits is straightforward. In the Schr{\"o}dinger picture, it results in
\begin{widetext}
  \begin{equation}
    \frac{d}{dt}\rho_{\text{s}}\left(t\right)
    =  \mathcal{L}\left(\rho_{\text{s}}\left(t\right)\right)=
    -i\left[H_{\text{s}},\rho_{\text{s}}\left(t\right)\right]+\gamma\sum_{i}\left(L_{i}\rho_{\text{s}}\left(t\right)L_{i}^{\dagger}-\frac{1}{2}\left\{ L_{i}^{\dagger}L_{i},\rho_{\text{s}}\left(t\right)\right\} \right)\label{eq:master_eq}
\,.
  \end{equation}
\end{widetext}

The idea of passive steering is to choose the operators $L_i$ such that repeated system-detector interactions in the long-time limit drive the system into the target state $\rho_{\oplus}$, which is often (but not necessarily) the ground state of the system Hamiltonian $H_{\text{s}}$. This is guaranteed if $\rho_{\oplus}$ is the unique steady state of $\mathcal{L}$ since then any initial state of the system is driven towards $\rho_{\oplus}$~\footnote{We assume here that $\mathcal L$ has no purely imaginary eigenvalues~\cite{yoshida2023SteadyStateUniqueness}}. Ideally, the target state is pure, $\rho_{\oplus} = \ketbra{\psi_\oplus}$. Passive steering is based on a predefined set of system-detector couplings, which, in contrast to active steering, is not modified in the course of the protocol based on the detector readouts. When the measurement outcomes are completely discarded (even for the termination of steering), such a passive protocol is sometimes also referred to as blind steering. In the context of measurement-based cooling, such protocols can be dubbed ``blind cooling''; this term will be utilized throughout the paper. 

\subsection{Choice of jump operators}\label{subsec:Lchoice}

In order to understand how to engineer the system-detector interactions, we determine which jump operators are permissible by asking that they leave the target state $\ket{\psi_\oplus}$ invariant. For simplicity, we will focus on one-dimensional lattices (``chains''), where the geometrical locality of the jump operators is established by involving neighborhoods of two adjacent sites which we call a \emph{link} $\left(i, i+1 \right )$. Our considerations are, however, not restricted to one-dimensional systems, and the generalization to larger neighborhoods is straightforward. The target state on a given link $\left(i, i+1 \right )$ is obtained by tracing over all other sites,
\begin{align}
\rho_{\oplus}^{\left(i,i+1\right)}= & \text{Tr}_{\mathcal{H}_{i}^{\complement}}\left(\rho_{\oplus}\right)\,,
\end{align}
where $\mathcal{H}_{i}^{\complement}\equiv\bigotimes_{j\neq i,i+1}\mathcal{H}_{j}$ denotes the tensor complement to the local Hilbert space of the link, $\mathcal{H}^{\left(i,i+1\right)}$.
We seek to cool (empty) all states on a given link $\left(i,i+1\right)$ that do not pertain to $\rho_{\oplus}^{\left(i,i+1\right)}$. 

As a simple example, consider a qubit with target state $\rho_\oplus=\ketbra{0}$. Since $\ket{1}$ does not pertain to $\rho_\oplus$, cooling via $L=\sigma^-=\ketbra{0}{1}$ would be permissible. We can formalize this intuition via the support of $\rho_\oplus$.
Defining the \emph{support} of a state $\rho_\oplus$ as 
\begin{equation}
    \text{supp}\,\rho = \left ( \ker \rho \right )^\perp,
\end{equation}
the local Hilbert space $\mathcal{H}^{\left(i,i+1\right)}$ of the link is partitioned into a hot and a cold local subspace, 
\begin{subequations}
    \begin{eqnarray}
        \mathcal{H}^{\left(i,i+1\right)} & =&\mathcal V_{\text{cold}}^{\left(i,i+1\right)}\oplus \mathcal V_{\text{hot}}^{\left(i,i+1\right)}\,,\\
        \mathcal V_{\text{cold}}^{\left(i,i+1\right)} & =&\text{supp}\,\rho_{\oplus}^{\left(i,i+1\right)}\,, \label{eq:v_cold_measurements}\\
        \mathcal V_{\text{hot}}^{\left(i,i+1\right)} & =&\left(\mathcal V_{\text{cold}}^{\left(i,i+1\right)}\right)^{\perp}=\ker \rho^{\left(i,i+1\right)}_\oplus \,.  \label{eq:v_hot_measurements}
    \end{eqnarray}
\end{subequations}
We refer to $\mathcal V_{\text{hot}}^{\left(i,i+1\right)}$ as \emph{hot} because it contains all the states which do not pertain to $\rho_{\oplus}^{\left(i,i+1\right)}$ and therefore cooling consists of emptying the hot subspace so that the system's state is in the \emph{cold} local subspace $\mathcal V_{\text{cold}}^{\left(i,i+1\right)}$.

The partitioning of the local Hilbert space of the cooled link into hot and cold subspaces leads to a natural way of choosing the jump operators: For every state $\ket{\phi_{\text{h},j}} \in \mathcal V_{\text{hot}}^{\left(i,i+1\right)}$, we choose some state $\ket{\phi_{\text{c},j}}\in\mathcal V_{\text{cold}}^{\left(i,i+1\right)}$ such that $L_j = \ketbra{\phi_{\text{c},j}}{\phi_{\text{h}, j}}$. This choice ensures that the jump operators $L_j$ do nothing but cool, i.e., they do not transfer population from the cold to the hot subspace or within the hot subspace, and they do not result in pure dephasing, but transfer population into the cold subspace. 

This intuition can be formalized in terms of three conditions. 
\begin{subequations}\label{eq:cond-jump}
First, the operators $L^{\left(i,i+1\right)}_j$ must be  nilpotent, i.e., 
\begin{equation}\label{eq:nilpotent}
    \left(L^{\left(i,i+1\right)}_j\right)^2=0\,.
\end{equation}
This ensures all eigenvalues of $L_j^{\left(i,i+1\right)}$ to be $0$. Otherwise, $L_j^{\left(i,i+1\right)}$ would possess at least one non-zero eigenvalue, implying existence of an invariant state.
The second property guarantees that only states from the cold subspace are not affected by any of the jump operators: 
\begin{equation}
    \cap_j \ker L_j^{\left(i,i+1\right)} = \mathcal V_{\text{cold}}^{\left(i,i+1\right)}\,.
\end{equation}
Moreover, 
\begin{equation}\label{eq:kerLj}
    \sum_j \text{image}\left(L^{\left(i,i+1\right)}_j \right) \subseteq \mathcal V_{\text{cold}}^{\left(i,i+1\right)},
\end{equation}
\end{subequations}
since then the jump operators only map states from the hot subspace to the cold subspace but do not act within the hot (respectively, cold) subspace.

\subsection{Cooling rate}\label{subsec:rate}

Designing the system-detector interaction that drives the system into the desired target state is an important first step. For practical applications of the protocol, the time that is required to attain the target state shall be faster than detrimental effects, that tend to heat up the system. This is also relevant, in particular for many-body systems. A key question is how this characteristic time scales with the system size.

To quantify the approach of the target state, a suitable figure of merit is needed. A natural choice is the Hilbert-Schmidt overlap,
\begin{equation}
    \label{eq:HSoverlap}
    D_{\text{overlap}} = \Tr\left( \rho_{\text{f}}\, \rho_\oplus\right)\,,
\end{equation}
between the obtained state $\rho_{\text{f}}$ and the target state $\rho_\oplus$. It is a suitable figure of merit as long as the target state is pure~\cite{BasilewitschAQT19}. This is the case here, and we will use it throughout. For mixed states, one would need to consider a true distance measure~\cite{BasilewitschAQT19} such as the Hilbert-Schmidt distance, $D_{\text{HS}}=\frac{1}{2}\Tr\left( \left(\rho_{\oplus}-\rho_{\text{f}}\right)^{2}\right)$. Since we are interested in cooling, one could also consider the expectation value of the system energy, relative to the targeted energy, i.e., $E^{rel}_s = \Tr\left( H_\text{s} \rho_\text{f}\right) - \Tr \left(H_\text{s} \rho_\oplus\right)$. 

The rate of cooling is set by the spectral gap of the Liouvillian, defined as $\Delta = \left|\Re \lambda_1\right|$ where $\lambda_1$ is the non-zero eigenvalue of $\mathcal{L}$ with the real part closest to $0$. The gap sets the slowest rate of approach towards $\rho_{0}$, which dominates in the asymptotic limit $t\rightarrow\infty$. The larger the gap, the faster the convergence towards $\rho_{\oplus}$.
If there exist multiple steady states, namely, there are multiple right eigenvectors at eigenvalue 0, then $\Delta=0$. In this case for some initial states the dynamics will not converge to $\rho_{\oplus}$, and blind cooling is not feasible.
Uniqueness of the steady state can be checked in terms of the algebraic properties of the Hamiltonian $H$ and the jump operators $L_i$~\cite{yoshida2023SteadyStateUniqueness}.

Obtaining $\Delta$ from exact diagonalization quickly becomes unfeasible,
since the Liouvillian corresponds to a $d^{2}\times d^{2}$-matrix
for a system Hilbert space dimension $d$. With the separation of
the Hilbert space into cold and hot subspaces introduced above in
Sec.~\ref{subsec:Lchoice}, and under the assumption of a unique
steady state, it is possible to estimate the gap $\Delta$ determining
the cooling rate for sufficiently small dissipation rates $\gamma$ as
\begin{equation}
\Delta\lesssim\Delta_{\text{est}},\label{eq:gap_est}
\end{equation}
with the estimate $\Delta_{\text{est}}$ deriving from the properties
of the Hamiltonian: \begin{subequations}\label{eq:gap_estimate}
\begin{eqnarray}
\Delta_{\text{est}} & \equiv & \frac{1}{2}Q\cdot\gamma\,,\\
Q & \equiv & \min_{\epsilon}q_{\epsilon}\,,\label{eq:q_epsilon}\\
q_{\epsilon} & \equiv & \min_{n}\expval{P_{\text{hot}}}{\epsilon_{n}}\,.
\end{eqnarray}
\end{subequations} Here, $\epsilon$ labels the (possibly degenerate)
eigenenergies of $H$ with the corresponding eigenspace spanned by
$\ket{\epsilon_{n}}$, where $n$ runs over the degenerate eigenstates
The projector $P_{\text{hot}}$ projects onto the ``hot'' subspace
on the link $\mathcal{V}_{\text{hot}}^{\left(1,2\right)}$, cf. Eq.~\eqref{eq:v_hot_measurements},
assuming that this is the only link that is cooled. With these definitions,
$Q$ in Eq.~\eqref{eq:gap_estimate} gives the minimum population
that an excited eigenstate $\ket{\epsilon_{n}}$ can have in the hot
subspace.

The estimate~\eqref{eq:gap_estimate} is justified for sufficiently
small dissipation rate $\gamma$ (cf.~Appendix~\ref{app:gap_estimate}),
\begin{equation}
\gamma\ll N/d^{N}\,,\label{eq:gamma_vs_gap}
\end{equation}
where $N$ is the number of sites and $d$ the Hilbert-space dimension
on a single site, and the Hamiltonian is of the form $H=\sum_{i}O_{i}$,
with the spectrum of operators $O_{i}$ (which can act on site $i$
and neighbouring sites) confined to $[0,1]$, so that $N/d^{N}$ is
the typical level spacing of the Hamiltonian.

In our numerical studies of the examples below, we find irrespective
of $\gamma$ that $\Delta_{\text{est}}$ provides a bound for the
gap,
\begin{equation}
\Delta\leq\Delta_{\text{est}}.\label{eq:gap_bound}
\end{equation}
We have not been able to prove this bound analytically, yet it looks
natural in the light of the following intuition: the cooling rate
is determined by the population of the hot subspace ($Q$) times the
dissipation rate ($\gamma$), yet may actually be smaller since after local
depopulation of the hot subspace excitations need time to propagate
to the cooled link; at the same time coherent superpositions between
the states of different energies evolve in time, as they acquire relative
phases --- this justifies the focus on the eigenstates of the Hamiltonian.

\section{Realization of single link cooling in two frustration-free spin chains}\label{sec:realization}

\subsection{Cooling into the AKLT state for $N\ge 5$ sites }
\label{subsec:aklt}

The ground state of the AKLT model~\cite{AKLT1,AKLT2} is a paradigmatic example of a matrix product state, as well as of a symmetry-protected topological phase~\cite{PollmannPRB2012}. In the context of quantum engineering, it holds promise as a resource state for measurement-based quantum computation~\cite{VerstraetePRA2004}. Measurement-assisted preparation of the AKLT ground state has recently attracted much attention~\cite{roy_measurement-induced_2020, HerasymenkoPRXQ2023, SmithPRXQ2023, chen2023highfidelity, PhysRevA.108.013712}, as an illustration of proof-of-principle protocols for engineering correlated many-body states.
We now show that measurement-induced cooling into the AKLT ground state can be taken to the dilute limit. 

\begin{figure}[tbp]
\includegraphics[width=1\columnwidth]{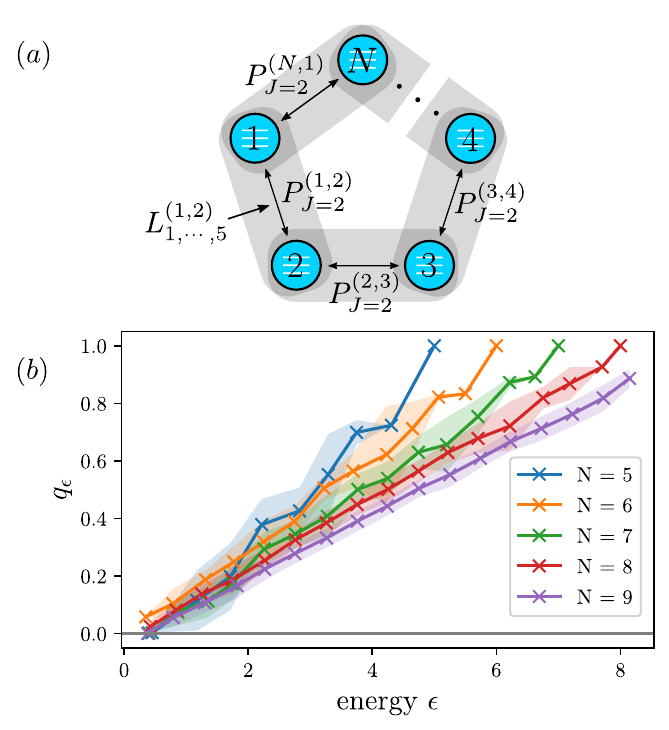}%
\caption{(a) Schematic of the spin-1 AKLT chain with periodic boundary conditions. The projectors onto the $J_i=2$ subspace for each nearest-neighbor link, $P_{J=2}^{(i,i+1)}$,  make up the AKLT Hamiltonian. In its most extreme form, dilute cooling acts on a single link only, via the five cooling operators $L^{(i,i+1)}_{1,\ldots,5}$, cf. Eq.~\eqref{eq:jumpAKLT}, with $i$ taken to be $i=1$. (b) Correlation of the minimal eigenstate projection onto the hot subspace $q_\epsilon$, Eq.~\eqref{eq:q_epsilon}, with energy $\epsilon$ for AKLT chains of increasing size. The solid lines show the mean of $q_\epsilon$ for energy bins of size $\Delta E = 0.5$ while the shaded background gives the standard deviation.  This correlation implies that in general low-lying excitations will decay slower than higher energy ones.}  \label{fig:aklt_model}
\end{figure}

The AKLT model~\cite{AKLT1, AKLT2}, a one-dimensional chain of spin-1 particles with nearest-neighbor interactions, cf. Fig.~\ref{fig:aklt_model}(a), is an example of a frustration-free system with Hamiltonian given by the sum of non-commuting local Hamiltonians.  Periodic boundary conditions (PBC) imply uniqueness of the ground state which is a valence bond state such that on each bond between two neighboring spins, there is no projection on the total spin-2 sector. The ground state of the whole chain is also the ground state of each local term of the ``parent Hamiltonian''.  
The latter property is what allows \cite{ticozzi_stabilizing_2012} for local steering, i.e., each jump operator needs to act on only one bond. 
Preparation of the AKLT ground state with passive steering was shown in Ref.~\cite{roy_measurement-induced_2020} for a total number of detectors that is extensive in system size. In that work, the passive steering protocol was employed to a system of \textit{non-interacting} spins-1. In the present work, the measurement-induced dynamics is superimposed onto the system dynamics governed by its own interacting Hamiltonian. We show that, remarkably, cooling a single link is sufficient to drive into the AKLT ground state for systems with $N\ge 5$ sites and then discuss the special case of chain sizes $N=3,4$.

With the nearest-neighbor interaction, two neighboring spin-1 particles can form pairs of total spins $J_{i}=0,1$ or $2$ where $J_{i}$ is the quantum number associated to the eigenvalues of $\left(\vec{J}_i\right)^2$, the square of the total spin operator $\vec{J}_{i}\equiv\vec{S}_{i}+\vec{S}_{i+1}$.
The AKLT Hamiltonian is defined as the sum of the projectors $P_{J=2}^{\left(i,i+1\right)}$ on the $J_{i}=2$ subspaces, thereby penalizing them energetically,
\begin{align}
H_{\text{AKLT}}= & \sum_{i=1}^{N}P_{2}^{\left(i,i+1\right)} \,.\label{eq:h_aklt}
\end{align}
To drive the system into the ground state of Eq.~\eqref{eq:h_aklt}, a possible choice of operators $L_{i}$ projects each $J_{i}=2$ state onto $J_{i}=0,1$ states~\cite{roy_measurement-induced_2020},
\begin{subequations}\label{eq:jumpAKLT}
  \begin{eqnarray}
L_{1}^{\left(i,i+1\right)} & =&\left(\ketbra{1,1}{2,2}\right)_{\left(i,i+1\right)}\,,\label{eq:jump_u1}\\
L_{2}^{\left(i,i+1\right)} & =&\left(\ketbra{1,1}{2,1}\right)_{\left(i,i+1\right)}\,,\label{eq:jump_u2}\\
L_{3}^{\left(i,i+1\right)} & =&\left(\ketbra{1,0}{2,0}\right)_{\left(i,i+1\right)}\,,\label{eq:jump_u3}\\
L_{4}^{\left(i,i+1\right)} & =&\left(\ketbra{1,-1}{2,-1}\right)_{\left(i,i+1\right)}\,,\label{eq:jump_u4}\\
L_{5}^{\left(i,i+1\right)} & =&\left(\ketbra{1,-1}{2,2}\right)_{\left(i,i+1\right)}\,,\label{eq:jump_u5}    
  \end{eqnarray}
\end{subequations}
where $i$ labels the sites, $\ket{J,m_J}_{\left(i,i+1 \right)}$ denotes the state of total spin $J_i$, formed from the spins on the sites $i, i+1$, with projection quantum number $m_J$. We note that the chosen jump operators are ``compatible'' with the Hamiltonian of the system in the sense that they do not affect the ground state of the Hamiltonian.

Figure~\ref{fig:aklt_monte_carlo} shows how the system state approaches the AKLT ground state as time evolves, for chain lengths up to $N=10$. Here, the five steering operators are applied to the link connecting sites 1 and 2.
The state overlap is calculated by averaging over Monte Carlo trajectories~\cite{montecarlo1} that unravel the GKSL equation~\eqref{eq:master_eq}, resampling 10000 trajectories 500 times with the bootstrap method~\footnote{The bootstrap resampling method \cite{bootsrapResampling} is useful to estimate quantities from a random distribution with low sample size. It also provides a simple way to assign confidence intervals and therefore standard deviations. The method works by repeatedly (here: 500 times) randomly resampling the data with replacement. Each of the resamples is used to estimate the quantities of interest and the final mean and variance are given by the resulting statistics.}.
Confidence intervals for Monte Carlo sampling are smaller than the linewidths of the curves. Figure~\ref{fig:aklt_monte_carlo} demonstrates an exponential approach of the AKLT ground state for all chain lengths $N$. The gap $\Delta$ is then easily extracted from the slopes in Fig.~\ref{fig:aklt_monte_carlo}. As one would expect, the approach of the ground state slows down with chain length $N$, with a clear difference between even and odd numbers of sites where the ground state is approached faster for even-length chains. This will be further analyzed below, in terms of the gaps. 

\begin{figure}[tbp]
\includegraphics[width=0.95\columnwidth]{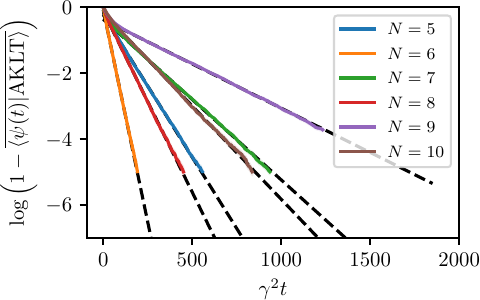}%
\caption{The logarithm of the overlap with the excited space as a function of time. We observe an exponential approach with time to the AKLT ground state $\ket{\text{AKLT}}$ for weak cooling ($\gamma = 0.1$) on a single link for several chain lengths $N$. The overlap is obtained from solving Eq.~\eqref{eq:master_eq} with the Monte Carlo wave-function method (with the bar denoting the average over the Monte Carlo trajectories). Dashed lines are linear fits and their slope is used to extract the Liouvillian gaps. 
  \label{fig:aklt_monte_carlo}}
\end{figure}

The overall possibility of dilute cooling and its dependence on system size can be rationalized as follows. Dilute cooling is the result of an interplay between local cooling and coherent non-local dynamics generated by the interactions between neighboring sites in the Hamiltonian. While excitations of the chain that are distant from the cooled link are not directly exposed to the jump operators, the interactions propagate them through the chain until they eventually reach the cooled link. Here, cooling takes place, taking energy out of the system by projecting the state onto an energetically lower subspace (from $J=2$ to $J=0,1$). If the state after the jump is the target state, then nothing more happens since the target state is left invariant by both the jump operators and the interactions. Otherwise, there still exists an excitation somewhere along the chain, and the procedure starts over by propagating the excitation through the chain until it reaches the cooled link. 

The interplay between dissipation and interaction-induced coherent dynamics naturally introduces two timescales: The dissipative timescale, given by the inverse local cooling rate, $1/\gamma$, is independent of the system size, whereas the time an excitation needs to propagate to the cooled link scales with system size. It is this scaling that explains the slowing down of the approach to the AKLT state observed in Fig.~\ref{fig:aklt_monte_carlo}. 

The slower approach of the AKLT state with increasing system size raises the question as to whether the target state can be attained in the thermodynamic limit. This is determined by the scaling of the Liouvillian gap $\Delta$ with the chain length. For our finite-size numerics, we have extracted the gaps from the linear fits in Fig.~\ref{fig:aklt_monte_carlo} and compare them to the Hamiltonian estimates~\eqref{eq:gap_bound} in  Fig.~\ref{fig:aklt_gap_extracted}. We find a finite gap $\Delta$ for chain sizes $N\ge 5$, indicating that cooling a single link drives the system towards $\ket{\text{AKLT}}$ for $t\rightarrow\infty$. The gap $\Delta$ as a function of system size roughly follows a power law for both even and odd $N$ but with different exponents $\alpha$, cf. Fig.~\ref{fig:aklt_gap_extracted}. While small chains with even length yield a larger $\Delta$ than odd chains of comparable lengths, Fig.~\ref{fig:aklt_gap_extracted} also shows that $\Delta$ decreases faster with growing even $N$ than it does for odd $N$. 

\begin{figure}[tbp]
\includegraphics[width=0.95\columnwidth]{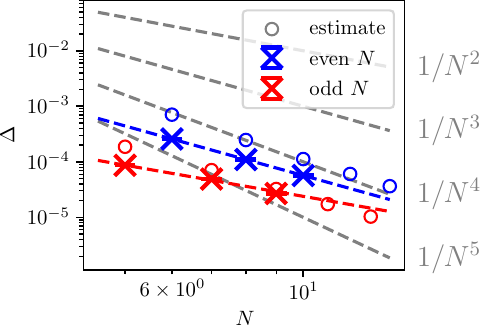}%
\caption{Liouvillian gap as a function of the AKLT chain size. Gaps (with error bars) extracted from the numerical overlaps in Fig.~\ref{fig:aklt_monte_carlo}
  are indicated with an "x" marker, gap estimates from Eq.~\eqref{eq:gap_estimate} with an "o" marker.
  A clear even-odd effect with respect to the chain length is observed, with $\Delta$ roughly following a separate power law in the two cases. Fitting to $\Delta(N)=m \cdot N^{-\alpha}$ yields exponents $\alpha_{\text{even}} = 2.97$, $\alpha_{\text{odd}} = 1.87$ and y-intersects $m_{\text{even}} = 52.3 \cdot 10^{-2}$, $m_{\text{odd}} = 17.5 \cdot 10^{-3}$. \label{fig:aklt_gap_extracted}}
\end{figure}

Chains of size $N=3,4$ turn out to be a special case in which single-link cooling fails. This is due to the existence of dark states: excited eigenstates which are also steady. These states are not subject to cooling on link 1-2 because on that link they live within the $J_1=0,1$ subspace and thus have no overlap with the local hot subspace. For $N=3$ there is only a single link with a well-defined spin as the other spin operators $\vec{J}_{2}$ and $\vec{J}_{3}$ do not commute with $\vec{J}_{1}$. The {\it dark states} all show $J_1=1$. In a chain with $N=4$ sites, there can be two links with well-defined spins, e.g., $\left(1,2\right)$ and $\left(3,4\right)$. The {\it dark states} turn out to feature $J_1=0$ and $J_3=2$. 
Starting with $N=5$, we find no {\it dark states that are steady states}. With periodic boundary conditions, $N=5$ is the smallest chain where any link with the well-defined total spin (meaning that the spin operator of the link in question commutes with $J_1$) has at least one other link with a well-defined spin between itself and the cooled link. 
Note that chains of size $N=3,4$ can be cooled in a non-dilute way, when coupling all links to detector qubits. The failure to cool chains of lengths  $N=3$ and 4 with steering a single link can be remedied by adding coherent local interactions. We first identify the general conditions under which dilute cooling is possible in Sec.~\ref{sec:cooling_conditions} and come back to the AKLT chain with $N=3,4$ in Sec.~\ref{subsec:smallAKLT}, showing that nearest-neighbor interactions are sufficient.

\subsection{Selective preparation of a degenerate ground state of the  Majumdar-Ghosh model}
\label{subsec:majumdar_ghosh}

Dilute cooling is, of course, not limited to the AKLT model. We show now that it works equally well for another frustration-free spin chain with Hamiltonian given as sum over non-commuting terms, the Majumdar-Ghosh (MG) model \cite{MajumdarGhoshOriginal}. Here, three neighboring spin-$1/2$ particles are coupled into a total spin with quantum numbers $J_i = 1/2,\, 3/2$ associated to $\left(\vec{J}_i \right)^2$ where $\vec{J}_{i} = \vec{S}_{i} + \vec{S}_{i+1} + \vec{S}_{i+2}$. The original formulation of the MG Hamiltonian was $H_\text{MG} = \sum_j^N \vec{J_i}^2$. Up to a shift in energy, the model can be formulated in terms of projectors such that the $J_i = 3/2$ subspaces are energetically penalized by the projector $P^{\left(i\right)}_{J=3/2}$ acting on sites $i,i+1,i+2$ via the Hamiltonian
\begin{equation}
    H_{\text{MG}} = 12\sum_{i=1}^{N} P^{\left(i\right)}_{J=3/2} \,,\label{eq:majumdar_ghosh}
\end{equation}
subject to PBC. We consider chains with an even number of sites $N$ because the ground state manifold is only two-fold degenerate in this case. It is spanned by the states $\ket{\psi_\pm}$---the product states of spin singlets formed from two neighboring spins with the first spin on an even ($-$), respectively odd ($+$), site,
\begin{align}
\ket{\psi_{-}}= & \bigotimes_{i=1}^{N/2}\ket{0,0}_{\left(i,i+1\right)}\,,\\
\ket{\psi_{+}}= & \bigotimes_{i=1}^{N/2}\ket{0,0}_{\left(i+1,i+2\right)}\,,
\end{align}
where $\ket{0,0}_{\left(i,i+1\right)}$ denotes the $S=0, m=0$ singlet state on link $(i, i+1)$.
Although the ground state is two-fold degenerate, it is possible to select one of them with dilute cooling and we choose $\ket{\psi_-}$ as the target ground state. 

The ground state degeneracy represents the main difference between the MG and AKLT models which needs to be accounted for when designing a cooling protocol targeting a pure target state. 
Naive adaptation of the recipe for choosing the jump operators presented in Sec.~\ref{subsec:Lchoice}, considering three neighboring sites of the MG model, suggests projecting every $J=3/2$ state onto some $J=1/2$ state. This does not, however, lead to a pure steady state but also allows for statistical mixtures of $\ket{\psi_\pm}$. To single out one of the two ground states $\ket{\psi_\pm}$, it is sufficient to perform cooling on only two neighboring sites. Two neighboring $S=1/2$ spins combine into either $J=1$ or $J=0$. We choose to cool the link $(1,2)$. The $\ket{\psi_-}$ state lives within the $J=0$ subspace on this link.
A suitable choice of operators that act on the single link $(1,2)$ transforming the states from the triplet manifold to the singlet manifold is then given by 
\begin{subequations}
\begin{align}
L_{1}= & \left(\ketbra{0,0}{1,1}\right)_{\left(1,2\right)}\,,\\
L_{2}= & \left(\ketbra{0,0}{1,-1}\right)_{\left(1,2\right)}\,,\\
L_{3}= & \left(\ketbra{0,0}{1,0}\right)_{\left(1,2\right)}\,,
\end{align}
\end{subequations}
This implies 
\[ V^{\left(1,2\right)}_{\text{hot}} = \text{span} \{ \left(\ket{1,1}\right)_{\left(1,2\right)}, \left(\ket{1,0}\right)_{\left(1,2\right)}, \left(\ket{1,-1}\right)_{\left(1,2\right)} \}
\] 
and $V^{\left(1,2\right)}_{\text{cold}} = \left(\ket{0,0}\right)_{\left(1,2\right)}$. We therefore choose $L_i$ such that any state from $V^{\left(1,2\right)}_{\text{hot}}$ is projected into $V^{\left(1,2\right)}_{\text{cold}}$. Again in analogy to the AKLT model, excitations outside of the cooled link are propagated through the chain by the interactions, until they eventually reach link $(1,2)$ where cooling takes place.

\begin{figure}[tbp]
\includegraphics[width=0.95\columnwidth]{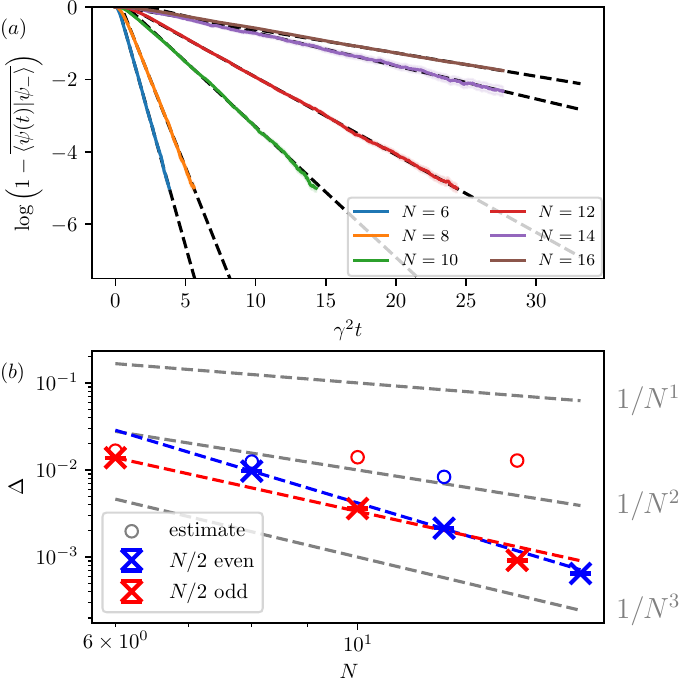}%
\caption{Scaling of the cooling rate with chain size $N$ for the MG model. (a) Overlap with the target state $\ket{\psi_{-}}$, bootstrap-averaged 500 times over 10000 trajectories. Dashed lines are linear fits used to extract $\Delta$. (b) The gap fitted to a power law $\Delta(N)=m\cdot N^{-\alpha}$. Gaps (with error bars) extracted from (a) are indicated with an "x" marker, gap estimates from Eq. ~\eqref{eq:gap_estimate} with an "o" marker, with exponents $\alpha_{N/2 \,\text{even}} =3.76 $, $N_{N/2\,\text{odd}} = 2.78$. The y-intersects are $m_{N/2 \text{ even}} = 24.4$ and $m_{N/2 \text{ odd}} = 2.02$. \label{fig:mg_overlap}}
\end{figure}

The approach to the target state is shown in Fig.~\ref{fig:mg_overlap}(a) for chains with an even number of sites up to $N=16$, where we have used the same numerical method as in Sec.~\ref{subsec:aklt}. The target state is approached exponentially in time for chains of length $N \geq 6$ with the exponent, i.e., the slope in Fig.~\ref{fig:mg_overlap}(a), which depends on $N$ differently for chains with even and odd values of $N/2$. This difference is more clearly seen in Fig.~\ref{fig:mg_overlap}(b) plotting the gap obtained from the slopes in Fig.~\ref{fig:mg_overlap}(a).  The gap estimate~\eqref{eq:gap_estimate} gives a fairly good prediction of $\Delta$ for system size up to $N=10$ in Fig.~\ref{fig:mg_overlap}. 

As explained in Appendix~\ref{app:gap_estimate}, the estimate can only be expected to give a tight bound if the difference in eigenenergies $\delta \epsilon = |\epsilon - \epsilon'|$ for any pair of energies $\epsilon, \epsilon'$ is large compared to $\gamma$ but for chains with $N\geq 10$, this is no longer the case. The gap roughly follows a power law similar to the AKLT model but the size of the gap is larger by about two orders of magnitude compared to the AKLT chains of similar size, cf. Fig.~\ref{fig:aklt_gap_extracted}.
This is in line with the overlaps $Q$, cf. Eq.~\eqref{eq:q_epsilon}, also being roughly two orders of magnitude larger compared to the AKLT model. It is thus easier to scale up the system size in the dilute cooling of the MG model, in the sense that for similar system sizes the target state will be reached much faster in the MG model as compared to the AKLT state.

\section{Conditions for dilute cooling \label{sec:cooling_conditions}}

Detector qubits are a precious resource. In contrast, the ability to engineer quasi-local coherent interactions is a basic prerequisite for the realization of any measurement-induced dynamics --- the system needs to be coupled to the detector qubits in the first place. It should then also be comparatively straightforward to alter the Hamiltonian by adding extra terms, as long as they are quasi-local, see also Fig.~\ref{fig:schematic_cooling}. 

This defines a paradigm that allows us to remedy the failure of dilute cooling protocols utilizing only the measurement-induced jump operators, as in the example of the AKLT chain of length $N=3,4$. We formulate a general framework to leverage quasi-local interactions to overcome this problem. To this end, we identify necessary and sufficient conditions for dilute cooling:  Assuming full quasi-local coherent control over the system (in the sense that we are free to add to the system Hamiltonian any quasi-local operator), under which 
conditions becomes a desired target state $\rho_\oplus=\ketbra{\psi_\oplus}$  the unique steady state of $\mathcal{L}$? In other words, the answer to the question of whether a given $\ket{\psi_\oplus}$ is obtainable via blind dilute cooling does then not depend on $H$ implementing $\ket{\psi_\oplus}$ as its ground state.

\subsection{Permissible quasi-local coherent interactions}\label{subsec:Hchoice}

We restrict ourselves to the case of nearest-neighbor interactions but longer-ranged interactions or more complicated neighborhood definitions are also possible. To set the stage, we first define the space of allowed additional coherent terms. They have to fulfill two conditions: (i) they should be quasi-local, and (ii) they need to leave the target state $\ket{\psi_\oplus}$ invariant. A suitable operator basis to represent them is the generalized Pauli basis $\{\sigma^{i}_j\}$, where $j\leq d^2$ with $d$ the dimension of the local Hilbert space on each site $i$. The space of all quasi-local operators can then be written as
\begin{align*}    \mathcal{O}&=\text{span}\cup_{i=1}^{N}\mathcal{O}^{\left(i,i+1\right)}\,, 
\\
\mathcal{O}^{\left(i,i+1\right)}&=\text{span}\left\{ \sigma_{j}^{\left(i\right)}\otimes\sigma_{k}^{\left(i+1\right)}:j,k\leq d^{2}\right\} \,.
\end{align*}
From $\mathcal{O}$, we select those operators that leave $\ket{\psi_\oplus}$ invariant. In the most general sense, this requires $\ket{\psi_\oplus}$ to be an eigenstate to all the selected operators. In order to obtain a vector space structure, we only allow for selected operators $A\in \mathcal{O}$, such that $\ket{\psi_\oplus}$ is eigenstate to $A$ with eigenvalue $0$. We thus define 
\begin{equation}
    K_{\ket{\psi_{\oplus}}}	\equiv\left\{ A\in\text{span }\mathcal{O}:A\ket{\psi_{\oplus}}=0\right\} \oplus \{ \mathbf{1}\} \,, \label{eq:def_kernelizer}
  \end{equation}
dubbing the space of allowed quasi-local coherent terms appropriately the \emph{kernelizer}. Note that we do not lose any generality by restricting to annihilators of $\ket{\psi_\oplus}$ since operators $A\in\mathcal{O}$ with $A\ket{\psi_\oplus} = \lambda \ket{\psi_\oplus}$ and $\lambda \neq 0$ can be obtained by simply adding the identity to $K_{\ket{\psi_{\oplus}}}$.

The key property of $K_{\ket{\psi_{\oplus}}}$ as defined in Eq.~\eqref{eq:def_kernelizer} is to separate the system Hilbert space $\mathcal{H}$ into the target space $\mathcal{H}_{\oplus}=\text{span}\left\{ \ket{\psi_{\oplus}}\right\}$ and its complement $\mathcal{H}_\oplus^\complement$. Our definition~\eqref{eq:def_kernelizer} comes with the advantage that $K_{\ket{\psi_{\oplus}}}$ has a vector space structure over $\mathds{R}$ and contains only Hermitian operators. The kernelizer is thus spanned by all coherent interactions affecting only $H_\oplus^\complement$. The purpose of the separation of Hilbert space is two-fold: (i) it allows for identifying the operators that leave the target state invariant and (ii) it implies a straightforward sufficient condition for blind cooling in terms of operator controllability on the complement $H_\oplus^\complement$, see Sec.~\ref{subsec:suffcond} below.

\subsection{Necessary conditions} \label{subsec:necessarycond}
We identify two necessary conditions for dilute cooling, one on the jump operators and one on the coherent interactions.
The condition on the jump operators formalizes the intuition that cooling should not affect the target state, whereas the condition on the coherent interactions is required to ensure population flow to the cooled link. We state the conditions here and prove them in Appendix~\ref{app:necessarycond}.

The condition on the jump operators can be stated as follows:
\begin{subequations}\label{eq:neccond_meas}
Cooling on a given link $\left(i, i+1 \right)$ is permissible if the local hot subspace is not empty,
\begin{equation}
  \mathcal  V_{\text{hot}}^{\left(i,i+1\right)}  \neq\emptyset\,. 
  \label{eq:v_hot_cond_1}   
\end{equation}
Equivalently, since $\mathcal{H}_{\left(i,i+1\right)} =\mathcal V_{\text{cold}}^{\left(i,i+1\right)}\oplus \mathcal V_{\text{hot}}^{\left(i,i+1\right)}$, cooling is permissible if the local cold subspace is a proper subset of the local Hilbert space of the link $(i,i+1)$,
\begin{equation}
 \mathcal  V_{\text{hot}}^{\left(i,i+1\right)}  \subsetneq \mathcal{H}^{\left(i,i+1\right)} \,.\label{eq:v_hot_cond_3}  
\end{equation}
\end{subequations}
Note that condition ~\eqref{eq:neccond_meas} is equivalent to Eq.~\eqref{eq:cond-jump}: If the hot subspace is not empty, one can always choose a set of jump operators which fulfill Eqs.~\eqref{eq:cond-jump}. Conversely, given a set of jump operators obeying conditions ~\eqref{eq:cond-jump}, there necessarily has to exist a finite hot subspace.

Assuming that condition~\eqref{eq:neccond_meas} is met, the jump operators $L_j^{(i,i+1)}$ should ideally be chosen as $L^{\left(i,i+1\right)}_j = \ketbra{\phi_{\text{c},j}}{\phi_{\text{h},j}}$ such that they map \emph{every} state in $\mathcal  V_{\text{hot}}^{\left(i,i+1\right)}$ to some states in $\mathcal  V_{\text{cold}}^{\left(i,i+1\right)}$, as suggested in Sec.~\ref{subsec:Lchoice}.
The choice of the $\{L_j^{(i,i+1)}\}$ is not unique. For example, the jump operators for the AKLT chain, Eq.~\eqref{eq:jumpAKLT}, map the five states with $J=2$  onto the four states with $J=0,1$ such that the sign of $J_z$ is preserved. Another of many conceivable choices would be $\{L_j^{(i,i+1)}\}$ that flip the sign of $J_z$.

A necessary condition on the coherent interactions can be stated in terms of the kernelizer.
Considering the coherent part of the equation of motion~\eqref{eq:master_eq}, a necessary condition for cooling is that there should be no state other than the target state that is invariant to both cooling and \emph{all} coherent interactions from the kernelizer. In particular, excitations on links that are not cooled have to propagate through the system until they reach the cooled link. We formalize this condition as
\begin{multline}
\nexists\ket{\phi}\in\mathcal{H}\setminus\left\{ \ket{\psi_{\oplus}}\oplus \left( \mathcal{V}_{\text{hot}}^{\left(i,i+1\right)}\bigotimes_{j\neq i}\mathcal{H}_{j} \right)\right\} \text{ such that } \\
\ket{\phi}\text{ is eigenstate of \emph{all} }A\in K_{\ket{\psi_{\oplus}}} \label{eq:necessary_condition}.
\end{multline}

The kernelizer contains \textit{all} Hermitian operators that leave the target state invariant, not only those that are part of the system Hamiltonian. This is why condition~\eqref{eq:necessary_condition} is necessary but not sufficient.

Condition~\eqref{eq:necessary_condition} explains why certain states, for example the highly entangled GHZ and W states,
\begin{align}
    \ket{\psi_\text{GHZ}}&=\frac{1}{\sqrt{2}}\left(\ket{0}^{\otimes N}+\ket{1}^{\otimes N}\right), \\
    \ket{\text{W}}&=\frac{1}{\sqrt{N}}\left(\ket{100\ldots0}+\ket{010\ldots0}+\ldots+\ket{00\ldots01}\right)
\end{align}
cannot be obtained by dilute cooling. For these states and $N\ge 5$, one can identify a complementary state that is not locally distinguishable from $\ket{\psi_\text{GHZ}}$, respectively $\ket{\text{W}}$, namely the antisymmetric GHZ state,  $\ket{\phi_\text{GHZ}}=\frac{1}{\sqrt{2}}\left(\ket{0}^{\otimes N}-\ket{1}^{\otimes N}\right)$ and the state $\ket{\text{W}_0}=\ket{000\ldots 0}$, respectively. Indeed, the reduced states of $\ket{\psi_{\text{GHZ}}}$ and $\ket{\phi_{\text{GHZ}}}$ on any link are identical which implies that there is no quasi-local operator, neither as part of the Hamiltonian nor the jump operators, that can discriminate between the two states. In contrast, with less than five qubits, condition~\eqref{eq:necessary_condition} is satisfied, and dilute cooling is possible. Given that a link involves two qubits, our argument is in line with the general proof~\cite{ticozzi_steady-state_2014} that cooling towards a GHZ state requires measurements on at least $N/2$ qubits~\footnote{To be precise, non-dilute cooling towards a GHZ state requires conditional stabilization, i.e., active steering~\cite{ticozzi_steady-state_2014}. Passive steering is possible when targeting a GHZ state only approximately~\cite{martin2023}}. Beyond the framework of Hamiltonians and jump operators with quasi-local interactions that we consider here, dissipative preparation of a GHZ can also be achieved with a single auxiliary degree of freedom, provided it is coupled globally to all sites~\cite{Morigi:2015}.

\subsection{Sufficient condition}\label{subsec:suffcond}

The intuition that the role of the quasi-local coherent interactions is to ensure population flow towards the cooled link allows us to formulate a sufficient condition for dilute cooling. 
The desired population flow can surely be realized if {\it any} unitary evolution on the complement Hilbert space $\mathcal{H}_\oplus^\complement$ is possible. The latter is guaranteed if the system part that is defined on $\mathcal{H}_\oplus^\complement$ is controllable, and controllability can be checked via the rank of the dynamical Lie algebra~\cite{dalessandro_introduction_2007}. In our case, what matters is the Lie Algebra $\mathfrak{L}\left(K_{\ket{\psi_{\oplus}}} \right)$ generated by the kernelizer. The corresponding Lie rank condition reads
\begin{equation}
    \dim\mathfrak{L}\left(K_{\ket{\psi_{\oplus}}} \right)	=\left(d^{N}-1\right)^{2}  - 1\,.\label{eq:sufficient_condition}
\end{equation}
If Eq.~\eqref{eq:sufficient_condition} holds, then there exist, in principle, sufficiently many interactions leaving the target state invariant such that any time evolution in $\mathcal{H}_\oplus^\complement$ can be generated.
Since it is a condition purely on the \emph{coherent} part of the dynamics, the necessary condition on the dissipative part~\eqref{eq:neccond_meas} must also be fulfilled. We note that this condition is sufficient: Controllability is a strong assumption and may not be necessary for cooling to work. But if condition~\eqref{eq:sufficient_condition} is met, then it is possible to add to the Hamiltonian operators from the kernelizer that induce the desired propagation of excitations through the chain towards the cooled link. Thus, Eq.~\eqref{eq:sufficient_condition} together with Eq.~\eqref{eq:neccond_meas} indeed ensures cooling towards $\rho_\oplus$.  

To illustrate the relationship between population flow and controllability, we verify in Appendix~\ref{app:sufficient_condition} that the AKLT model~\eqref{eq:h_aklt} fulfills the condition~\eqref{eq:sufficient_condition} for unitary controllability in the complement to the target subspace. In fact, it does so for any $N$. This implies that blind cooling should be possible also for chains of size $N=3,4$. In other words, condition~\eqref{eq:sufficient_condition} tells us how to amend the AKLT Hamiltonian such that blind cooling becomes possible also for the small chains.

\subsection{Use of the sufficient condition to design single link cooling for the AKLT chain with  $N=3,4$ sites }\label{subsec:smallAKLT}

For chains of size $N=3,4$ and cooling on a single link, multiple steady states exist for the AKLT model. However, since the sufficient condition~\eqref{eq:sufficient_condition} is satisfied for any $N$ (see Appendix~\ref{app:sufficient_condition}), there exists a coherent interaction $\delta H$ in the kernelizer, which renders $\ket{\text{AKLT}}$ the unique steady state. One possible choice is to add
\begin{subequations}\label{eq:aklt_s2_x}   
\begin{align}
    J_{2,x}^{\left(i,i+1\right)}=&\sum_{-2\leq m_{j}\leq2}\left(\ketbra{2,m_{j}}{2,m_{j}+1}\right)_{\left(i,i+1\right)}+h.c. 
\end{align}
on all but one link, i.e.,
\begin{equation}
    \delta H = \alpha \sum_{i=1}^{N-1} J_{2,x}^{\left(i,i+1\right)}\,.
\end{equation}
\end{subequations}
Since $J^{\left(i,i+1\right)}_{2,x}$ only acts on the local $J=2$ subspace, it is an element of the kernelizer, i.e., permissible. Its specific choice can be motivated as follows. We seek the generator of an evolution that mixes all the excited states to guarantee propagation towards the cooled link. The additional coherent interaction~\eqref{eq:aklt_s2_x} mimics the $J_x$ action on the $J=2$ subspace, generating rotations of the spin around the $x$ axis. 

Given our choice of $\delta H$, it will now depend on the interaction strength $\alpha$ and the dissipation rate $\gamma$ whether and how quickly the $\ket{\text{AKLT}}$ is reached. The rate of approach is determined by the Liouvillian gap, shown in Fig.~\ref{fig:aklt_N=3} for $N=3$ as a function of $\alpha$ and $\gamma$. A comparison to the gap estimate~\eqref{eq:gap_estimate}, obtained via exact diagonalization of $H_\text{AKLT}$, is also presented in Fig.~\ref{fig:aklt_N=3}(b). In the weak cooling limit ($\gamma \ll 1$), in which Eq.~\eqref{eq:master_eq} is valid, the estimate captures the gap very accurately. 
When the coupling to the detector becomes comparable or larger than the nearest-neighbor interactions, one would need to solve the full system-detector dynamics~\cite{puente2023}. In this regime, the gap is expected to be suppressed due to the quantum Zeno effect which freezes the coherent evolution. 
If the dissipation is strong compared to the interactions, the dissipation immediately removes any change to the state due to the interactions, before propagation through the chain can happen. Even in the weak coupling limit, Fig.~\ref{fig:aklt_N=3}(a) suggests maximizing $\Delta$ by proper choice of the relative strengths of the additional coherent interaction and coupling to the detector qubit. More generally, one may also optimize the choice of the additional interactions $\delta H$ out of all the operators contained in the kernelizer $K_{\ket{\psi_{\oplus}}}$.
For $N=4$, the AKLT model behaves in much the same way, i.e., dilute cooling becomes possible by adding coherent interactions.

\begin{figure}[t!]
\includegraphics[width=1\columnwidth]{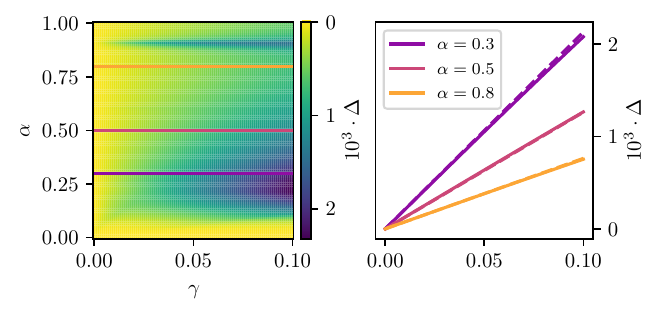}
\caption{Liouvillian gap of the AKLT chain for $N=3$ with the additional coherent interaction~\eqref{eq:aklt_s2_x} on links $(1,2)$ and $(2,3)$ as a function (a) of coherent interaction strength $\alpha$ and dissipation rate $\gamma$ and (b) of $\gamma$ for three values of $\alpha$ (solid lines). The gap estimates, cf. Eq.~\eqref{eq:gap_estimate}, are also shown (dashed lines). The solid lines in (a) indicate the values of $\alpha$ shown in (b). \label{fig:aklt_N=3}
}
\end{figure}

More generally, these observations raise the question of how to choose $\delta H$. As already mentioned above, this may be viewed as an optimization problem on the kernelizer space, targeting maximum $\Delta$ or at least a non-zero gap. Unfortunately, the kernelizer space is high-dimensional, and the simplest way to approach the question is to randomly sample $\delta H \in K_{\ket{\psi_{\oplus}}}$ and then calculate $\Delta$. We have found in all of $1000$ samples that such random sampling yields $\delta H$ that results in successful cooling. We are thus led to conjecture that the set of $H \in K_{\ket{\psi_{\oplus}}}$ not leading to dilute cooling is sparse, possibly of measure zero. This is in line with earlier observations in the context of system size-extensive (i.e., non-dilute) steering~\cite{ticozzi_steady-state_2014}.

\section{Discussion: General framework for dilute measurement-induced cooling}
\label{sec:disc}

After showcasing dilute cooling for two paradigmatic many-body models with finite-size numerics, we now turn to the overall applicability of our approach. First, we explore how far one can push the system size in the dilute cooling of many-body systems, given that the approach of the target state slows down as the system size increases. Second, based on a summary of the general framework of dilute cooling, we discuss to which systems it can be applied.

The success of dilute cooling as the size of a many-body system increases
depends both on the type and on the strength of the interactions, within the system and of the system with the meter. The interaction strengths give rise to competing timescales, whose scaling with system size is estimated in the following with the help of a classical diffusion picture. This can readily be motivated by viewing the jump operators as implementing a localized sink for excitations. A classical picture obviously misses many subtleties, for example, the difference in the propagation speed for classical and non-classical correlations~\cite{Roy:2023}, as well as the possibility of (partial) integrability of the model~\cite{Moudgalya2018}. Therefore, the scaling arguments presented below will only give a rough estimate for generic correlated systems.

Consider a classical one-dimensional diffusion equation with a $\delta$-function source (or rather sink). Denoting the diffusion constant by $D$ and the sink strength by $A$, the following types of relaxation behavior for a system of size $L$, characterizing its spatial extent, can be expected (see Supplemental material~\cite{SuppMat} for details of the derivation):
\begin{itemize}
    \item the “ballistic regime” --- for small systems, $L \ll \ell$, where $\ell$ is the mean free path with respect to the interaction of elementary excitations,  the dissipation rate is determined by the measurement rate $\gamma$ and the many-body level spacing;
    \item the intermediate diffusion regime --- for $\ell \ll L \ll D/A$, the decay is exponential with decay time $t_A = L/A$, where $A = \ell\gamma$ and $D$ itself is determined by $\ell$ (note that the diffusion constant does not enter the decay rate explicitely in this regime);
    \item for large systems, $L \gg D/A$, the time dependence of the survival probability is first given by $t^{-1/2}$, but becomes exponential at $t\gg t_D = L^2/D$ with the `gap' given by $1/t_D$. In the limit of an infinite system, excitations decay in a power-law (more precisely, square-root) manner, implying zero gap for $L = \infty$, as it should be.
\end{itemize}
For any finite system, all the above regimes correspond to a finite cooling time: for long times the removal of excitations from the system occurs exponentially. The numerical calculations presented in this work cover the ``ballistic regime'' of not-too-large systems. Successful cooling in such ballistic systems suggests that the sink satisfies the conditions for cooling in larger systems once diffusion sets in, which guarantees single-sink cooling in an infinite system in the limit of infinite times.       

The assumption of quasi-local interactions is key to the competition of timescales as discussed above: For extended systems with quasi-local interactions, the speed with which excitations can propagate to the cooled site(s) is necessarily finite. The picture changes for global interactions, as is the case for example in cavity QED~\cite{defenu:2023}, where the upper bound on propagation is set by the system size~\cite{Tran:2020}. Dilute cooling in finite time is then in principle possible even in the limit of infinite system size. For large systems with interactions characterized by light-cone propagation of excitations~\cite{Tran:2020}, the best strategy for dilute cooling will be to balance ``diluteness" with system size. For example, one can conceive a protocol where only every tenth link is cooled such that the number of "meters" is still considerably smaller than the number of system sites.

Next, we provide a simple ``recipe" to identify systems amenable to dilute measurement-induced cooling, drawing directly on the examples of the AKLT and MG model:
\begin{itemize}\setlength{\itemsep}{0pt}
    \item Consider as target state a ground state of a frustration-free projector-based Hamiltonian; 
\item  This state can be stabilized (even without the system own Hamiltonian) by a frustration-free Lindbladian with local terms acting on all parts of the system;
\item The local jump operators of such a Lindbladian suggest the form of the local coupling of the system parts (cells) to the measurement apparatus (meter);   
\item  In the system with its own Hamiltonian,  add the coupling to a meter qubit inferred above only at a single elementary cell of the system (a link, i.e., two neighboring sites in the AKLT and MG examples above);
\item Unless the Hamiltonian has a perfectly flat band of excitations or there exist excitations that do not overlap with the chosen cell, 
dilute cooling is expected to work: for a finite system, all excitations will be essentially removed within a finite time.                      
\end{itemize}

The simple recipe can be generalized beyond frustration-free projector-based Hamiltonians, such as the AKLT and MG models, and can be made to work even in the case of flat bands or excitations avoiding the sink (as in $N=3,\, 4$ AKLT model). The kernelizer (Sec.~\ref{sec:cooling_conditions}) is key to this generalization. It allows for identifying Hamiltonians, Lindbladians, and possibly additional interactions to realize dilute cooling. Use of the kernelizer results in the following extension of the simple recipe:
\begin{itemize}\setlength{\itemsep}{0pt}
    \item Find a system with a ground state that fulfills the necessary condition~\eqref{eq:neccond_meas} and choose a set of quasi-local jump operators according to the recipe
    \item Check if the system fulfills the sufficient condition~\eqref{eq:sufficient_condition} on the kernelizer, i.e., the existence of sufficiently many Hermitian operators acting on the hot subspace only, such that their Lie algebra is of full rank. In other words, check if the system is completely controllable on the (global) hot subspace; 
    \item Add as many operators from the kernelizer to the Hamiltonian as needed to ensure a unique steady state.
\end{itemize}
The most difficult step is the first one, but earlier findings for non-dilute cooling with quasi-local couplings provide a good starting point~\cite{ticozzi_stabilizing_2012, ticozzi_steady-state_2014, JohnsonQIP2016}. Good candidates are, for example, systems with product entangled pair states~\cite{kraus_preparation_2008, Perez-GarciaQIC2008, verstraete_quantum_2009}, including matrix-product states for one-dimensional systems.

\section{Concluding Remarks}\label{sec:concl}
 
Using quantum measurements to engineer desired dissipative dynamics comes with the advantage of versatility and robustness but is also very costly. The scaling of resources for measurement-induced dynamics with system size is an important open but often overlooked question. We have addressed this question for measurement-induced cooling into many-body ground states. Using finite-size numerics and analytical arguments, we have shown that the number of detector qubits can under certain conditions be reduced (``diluted'') all the way to the extreme case of a single detector. 

In particular, we have demonstrated dilute cooling for two examples: the ground states of the AKLT and MG models. Both are frustration-free spin chains where the Hamiltonian is given by a sum over non-commuting quasi-local terms. In general, cooling or passive steering requires couplings within the system and with the detector that leave the target state invariant. For dilute cooling via weak measurements, this translates into necessary conditions on the jump operators and coherent interactions. To formalize these conditions, we have introduced two key concepts---the kernelizer and the local hot and cold subspaces. The latter are defined via the projection of the target state onto the local Hilbert space of the sites that are cooled. Invariance of the target state is ensured by the jump operators fulfilling Eq.~\eqref{eq:cond-jump} or, equivalently,  the local hot subspace not being empty. The kernelizer is the vector space of all quasi-local Hermitian operators that leave the target state invariant. Since these operators generate the flow of excitations towards the cooled site(s), a necessary condition is their mere existence. But the kernelizer is even more useful: It allowed us to state a sufficient condition for dilute cooling, in terms of controllability on the Hilbert space complement to the target state. Controllability of the kernelizer implies the existence of sufficiently many quasi-local Hermitian operators to generate any conceivable unitary evolution on the complement space. A constructive recipe to implement dilute cooling is then obtained simply by adding coherent interactions from the kernelizer to the Hamiltonian.

Our approach of using controllability to identify couplings within the system that enable dilute cooling is related to the algebraic approach of Ref.~\cite{yoshida2023SteadyStateUniqueness, zhang2023} to determine the operators that ensure uniqueness of the steady state. Combining the two should allow for an even larger class of many-body target states than discussed here.

To conclude, we have addressed the challenge of measurement-based cooling a many-body quantum system to its ground state without the costly burden of extensive measurements. By showcasing the effectiveness of cooling on a single link for frustration-free one-dimensional spin chains, we have demonstrated that passive steering protocols can be taken to the dilute limit under specific conditions. Our findings, supported by both analytical arguments and finite-size numerical simulations, emphasize the potential for reducing the resources needed for measurement-induced dynamics. Moving forward, these insights pave the way for further applications, offering a more cost-effective and versatile approach to state engineering and state manipulation in quantum systems. 
 
\begin{acknowledgments}
  The authors gratefully acknowledge financial support from the Deutsche Forschungsgemeinschaft (DFG), Project No. 277101999, CRC 183 (project B02), grants No. EG 96/13-1 and GO 1405/6-1, and Project-ID 429529648 -- TRR306 QuCoLiMa (``Quantum Cooperativity of Light and Matter'', project D02), from the German Federal Ministry for Education and Research (BMBF), Project No. 13N16201 NiQ, and the Israel Science Foundation. This research was supported in part by the National Science Foundation under Grants No. NSF PHY-1748958 and PHY-2309135. This work is part of HQI (\url{www.hqi.fr}) and BACQ initiatives and is supported by France 2030 program under the French National Research Agency grants with numbers ANR-22-PNCQ-0002 and ANR-22-QMET-0002 (MetriQs-France).
\end{acknowledgments}

\appendix

\section{Gap estimate \label{app:gap_estimate}}

Here we derive the estimate~\eqref{eq:gap_est} for small dissipation
rates $\gamma$. We do this by explicitly constructing an eigenstate of the Lindbladian that has eigenvalue $-i\omega-\kappa$ with $\kappa\approx\Delta_{\mathrm{est}}$.
Having an eigenstate with $\kappa\approx\Delta_{\mathrm{est}}$
guarantees that the Lindbladian gap does not exceed this value, i.e.,
$\Delta\lesssim\Delta_{\mathrm{est}}$.

Consider
\begin{equation}
\rho_{\mathrm{coh}}=\ket{\psi_{\mathrm{exc}}}\bra{\psi_{\oplus}},
\end{equation}
where $\ket{\psi_{\mathrm{exc}}}\in\mathcal{H}/\ket{\psi_{\oplus}}$
is any state in the excited subspace. $\rho_{\mathrm{coh}}$ represents
coherences between the target state and the excited subspace and will
decay as the population of the excited subspace gets eliminated. Under
Eq.~\eqref{eq:master_eq}, the time evolution of the coherences is governed by an effective Hamiltonian,
\begin{equation}
  \tilde{H}  =H_{s}-\epsilon_{\oplus}-i\frac{\gamma}{2}P_{\mathrm{hot}}.\label{eq:H_tilde}
\end{equation}
with
$\epsilon_{\oplus}$ the ground state energy of $\ket{\psi_{\oplus}}$
and 
\[
  i\frac{d}{dt}\rho_{\mathrm{coh}}  =\tilde{H}\rho_{\mathrm{coh}}\,,
\]
where we have used the identities $\bra{\psi_{\oplus}}L_{i}^{\dagger}=0$
and $\sum_{i}L_{i}^{\dagger}L_{i}=P_{\mathrm{hot}}$. Therefore, any
eigenstate $\ket{\lambda}$ of the non-Hermitian Hamiltonian $\tilde{H}$
restricted to the excited subspace gives rise to an eigenstate of
the Lindbladian governing the evolution~\eqref{eq:master_eq}:
\begin{equation}
\tilde{H}\ket{\lambda}=\lambda\ket{\lambda}\Rightarrow\mathcal{L}\left(\ket{\lambda}\bra{\psi_{\oplus}}\right)=-i\lambda\ket{\lambda}\bra{\psi_{\oplus}}.
\end{equation}

For $\gamma\rightarrow0$, the eigenstates of $\tilde{H}$ coincide
with the eigenstates $\ket{\epsilon}$ of $H_{s}$ up to corrections
$O(\gamma/\delta\epsilon)$, where $\delta\epsilon$ represents the level
spacing of $H_{s}$. Then, by virtue of first-order perturbation
theory, the respective decay eigenvalues are $\kappa=-\mathrm{Re}\left(-i\lambda\right)=-\mathrm{Im}\,\lambda=\frac{\gamma}{2}\expval{P_{\text{hot}}}{\epsilon}$
up to corrections $O(\gamma^{2}/\delta\epsilon)$. In the case of
degenerate eigenstates of $H_{s}$, one needs to diagonalize $P_{\mathrm{hot}}$
in each energy subspace in order to find the eigenstates of $\tilde{H}$.
The smallest $\kappa=\Delta_{\mathrm{est}}$ corresponds to the state
$\ket{\epsilon}$ with the smallest overlap with the hot subspace,
leading to Eq.~(\ref{eq:q_epsilon}).

\section{Proofs of the necessary conditions for dilute cooling}\label{app:necessarycond}

We provide here proofs of the necessary and sufficient conditions for dilute cooling stated in Sec.~\ref{subsec:necessarycond}.

A jump operator is permissible if it leaves the target state invariant in the sense that $L^{\left(i,i+1\right)}\otimes \mathds{1}^\complement \ket{\psi_\oplus} = 0$. Not all operators with this property lead to cooling which is why this is only a necessary condition. As a first step, we prove that a necessary condition for the existence of a non-trivial (i.e., non-zero) permissible jump operator on link $\left(i,i+1\right)$ is for the reduced density operator to not be full rank. This can be seen to be equivalent to the condition of a non-empty local hot subspace, Eq.~\eqref{eq:v_hot_cond_1}, by recalling  definition~\eqref{eq:v_hot_measurements} and using that any linear operator has full rank if and only if its kernel contains only $0$.

For convenience of notation we write $L=L^{\left(i,i+1\right)}$ assuming that the jump operator only acts on the cooled link $\left(i,i+1\right)$. \\
Claim: $L\otimes\mathds{1}^{\complement}\ket{\psi_{\oplus}}=0\Rightarrow\text{rk }\rho_{\oplus}^{\left(i,i+1\right)}<d^{2}$ \\
We prove the contraposition \[\text{rk }\rho_{\oplus}^{\left(i,i+1\right)}=d^{2}\Rightarrow L\otimes\mathds{1}^{\complement}\ket{\psi_{\oplus}}\neq0\,.\]
Consider the r.h.s.; it is equivalent to 
\begin{equation}\label{eq:Lcompl}
    L\otimes\mathds{1}^{\complement}\ketbra{\psi_{\oplus}}L^{\dagger}\otimes\mathds{1}^{\complement}\neq 0\,.
\end{equation}
Since the l.h.s. of Eq.~\eqref{eq:Lcompl} is a projector and using the cyclicity of the trace, \[\text{Tr}\left(L^{\dagger}L\otimes\mathds{1}^{\complement}\ketbra{\psi_{\oplus}}\right)\neq0\,.\]
Taking the partial trace over all but the cooled link allows us to rewrite Eq.~\eqref{eq:Lcompl} in terms of the reduced state, \[\text{Tr}\left(L\rho_{\oplus}^{\left(i,i+1\right)}L^{\dagger}\right)\neq0\,.\]
By assumption, $\rho_{\oplus}^{\left(i,i+1\right)}$ has full rank and we may write \[\rho_{\oplus}^{\left(i,i+1\right)}=\sum_{n=1}^{d^2}\lambda_{n}\ketbra{n}\] with $\lambda_{n}>0$ and $\left\{ \ket{n}\right\} $ an orthonormal basis of $\mathcal{H}_{i,i+1}$. Then the above trace becomes $\sum_{n}\lambda_{n}\text{Tr}\ketbra{\tilde{n}}$ with $\ket{\tilde{n}}\equiv L\ket{n}$. Therefore $\sum_{n}\lambda_{n}\text{Tr}\ketbra{\tilde{n}}>0$ unless $\ket{\tilde{n}}=0$ for all $n$. Since $L$ is required to be non-trivial, we can exclude this case which concludes the proof.

We now turn to the necessary condition on the coherent interactions, Eq.~\eqref{eq:necessary_condition}.
To grasp the intuition, consider a state $\ket{\phi}$ that is neither the target state nor in the hot subspace. Because it is not in the hot subspace, it cannot be affected by any jump operator (we had introduced the jump operators to only act on the hot subspace). Then $\ket{\phi}$ is not directly subject to cooling, and cooling can only be mediated through interactions. In order for this to work, there has to be some interaction $H$ from the kernelizer which connects $\ket{\phi}$ to some other state $\ket{\phi '}=H\ket{\phi}$. It may be that $\ket{\phi'}$ is then subject to cooling or that $\ket{\phi'}$ is connected via interactions to yet another state until eventually a state subject to cooling is reached. This observation can be phrased as a necessary condition on the interactions: For every such state $\ket{\phi}$ that is not directly subject to cooling there has to exist an operator from the kernelizer that maps $\ket{\phi}$ to some other state $\ket{\phi'}\neq \lambda \ket{\phi}$ for some $\lambda\in\mathds{C}$. We formalize this condition as
\begin{eqnarray*}
    \forall\ket{\phi}\in\mathcal{H}\setminus\left\{ \ket{\psi_{\oplus}}\oplus \left( \mathcal{V}_{\text{hot}}^{\left(i,i+1\right)}\bigotimes_{j\neq i}\mathcal{H}_{j} \right)\right\} \\
    \exists H\in K_{\ket{\psi_{\oplus}}}:\ket{\phi'}=H\ket{\phi}\neq\lambda\ket{\phi}\,.
\end{eqnarray*}
Negating both parts of the statement leads to Eq.~\eqref{eq:necessary_condition}: There can be no state $\ket{\phi}$ being neither the target nor in the hot subspace such that $\ket{\phi}$ would be an eigenstate to all operators from the kernelizer.

\section{Unitary controllability for the AKLT model\label{app:sufficient_condition}}

As an illustrative example, we show here that the sufficient condition for dilute cooling on the interactions, Eq.~\eqref{eq:sufficient_condition}, i.e., full unitary controllability on the complement of the target state, is satisfied for the AKLT chain. 
Our argument consists of three steps: First, we present the general structure of a subset of the kernelizer for the AKLT model. Although only a subset of the kernelizer, we show in step 3 that it generates an algebra of full rank on the complement space. Therefore control on this subset leads to full operator controllability on the complement space.
Second, we argue that the absence of non-trivial invariant subspaces of the Lie algebra generated by the kernelizer implies controllability. Finally, we present the proof that there can be no other invariant subspaces of the Lie algebra than the space complement to the target state and the target state itself.

First notice that by construction $\ket{\psi_{\text{AKLT}}}$ is orthogonal to the $J_i = 2$ sectors. Its reduced state on the cooled link, $\rho_{\text{AKLT}}^{\left(i,i+1\right)}$, only has support in the $J_i=0,1$ subspaces and may in general have $3+1$ non-zero eigenvalues. Moreover, $H_{\text{AKLT}}$ makes no distinction between all $J_i = 1$ states such that the eigenvalues of $\rho_{\text{AKLT}}^{\left(i,i+1\right)}$ are $3+1$-fold degenerate. Labeling them $\lambda_{1}=\lambda_{2}=\lambda_{3}\equiv\lambda$ and $\lambda_{0}$, we may write \[\rho_{\text{AKLT}}^{\left(i,i+1\right)}=\text{diag}\big(\underset{J_{i}=2}{\underbrace{0,0,0,0,0}},\underset{J_{i}=1}{\underbrace{\lambda,\lambda,\lambda}},\underset{J_{i}=0}{\underbrace{1-3\lambda}}\big)\,.\] Evidently, any action on the $J_i = 2$ subspace leaves $\rho_{\text{AKLT}}^{\left(i,i+1\right)}$ invariant. Therefore, the local kernelizer, $K_{\ket{\psi_{\text{AKLT}}}}^{\left(i,i+1\right)} \equiv K_{\ket{\psi_{\text{AKLT}}}}\cap \mathcal{O}^{\left(i,i+1\right)}$ where $\mathcal{O}^{(i,i+1)}$ is the space of operators on the link $(i,i+1)$, cf. Sec.~\ref{subsec:Hchoice}, is $25$-fold dimensional, containing all the generators of unitaries acting on the $J_i = 2$ subspace. In other words $K_{\ket{\psi_{\text{AKLT}}}}^{\left(i,i+1\right)} = \mathfrak{su}\left(5\right)$, the generator of $SU(5)$ acting on the $J_i=2$ subspace.

Since $\mathcal{O}$ is the space of all quasi-local interactions, we may express the (global) kernelizer as 
\begin{eqnarray*}
    K_{\ket{\psi_{\text{AKLT}}}} & =&K_{\ket{\psi_{\text{AKLT}}}}\cap \mathcal{O}\nonumber \\
 & =&K_{\ket{\psi_{\text{AKLT}}}}\cap\left(\mathcal{O}^{\left(1,2\right)}+\mathcal{O}^{\left(2,3\right)}+\ldots+\mathcal{O}^{\left(N,N+1\right)}\right).
\end{eqnarray*}
This implies \[K_{\ket{\psi_{\text{AKLT}}}}\cap \mathcal{O}^{\left(i,i+1\right)}\subseteq K_{\ket{\psi_{\text{AKLT}}}}\,,\] and we have \[\sum_{i=1}^{N}K_{\ket{\psi_{\text{AKLT}}}}^{\left(i,i+1\right)}\subseteq K_{\ket{\psi_{\text{AKLT}}}}\,.\] We proceed by showing that the Lie algebra generated by the sum of the local kernelizers is already sufficient to satisfy Eq.~\eqref{eq:sufficient_condition}. 
For brevity, we write $\mathfrak{L}\left(K_{\ket{\psi_{\oplus}}}\right)=\mathfrak{L}$ and $K_{\ket{\psi_{\oplus}}}=K$ in the reminder of the proof and argue that, if the Lie algebra $\mathfrak{L}$ acting on a vector space $V$ only possesses a single invariant subspace, this subspace is $V$ itself. Moreover in this case the dimension of $\mathfrak{L}$ is maximal, $\dim\mathfrak{L}=\dim V^{2}$. If it was smaller, then $\mathfrak{L}$ would not generate all unitaries on $V$ and there would be a subspace invariant under $\mathfrak{L}$. It is thus sufficient to prove that $\mathfrak{L}$ possesses exactly two invariant subspaces, $V\oplus\ket{\psi_{\text{AKLT}}}=\mathcal{H}$.
The invariance of the latter subspace follows by construction. In other words, it is sufficient to show that there exists no invariant subspace $V$ with dimension smaller than $\dim V=d^{N}-1$.

The strategy of the proof is to construct $V$ via repeated applications of $\mathfrak{L}$ on some initial state, chosen to be the all up state $\ket{v_{0}}\equiv\ket{1,\ldots,1}$ in the local spin-1 basis. For lack of an explicit expression for $\mathfrak{L}$, we will only consider actions of $K\subset\mathfrak{L}$ for the generation of $V$. We shall see that this is already sufficient. We introduce the notation for vector space orbits, $K\left[\mathcal{H}\right]=\text{span }\left\{ A\ket{\psi}:A\in K,\ket{\psi}\in\mathcal{H}\right\} $; it is the image of Hilbert space $\mathcal{H}$ (i.e., a vector space of states) under the vector space of operators $K$. As a first step, we prove the following
\begin{lemma}
There exists some $k\in\mathbb{N}$ such that $V\equiv V^{k}\equiv\mathfrak{L}\left[\ket{v_{0}}\right]$ is the smallest subspace invariant under $\mathfrak{L}$ and contains $\ket{v_{0}}$.
\end{lemma}

\begin{proof}
First, by definition $\mathbf{1}\in\mathfrak{L}$ such that $\ket{v_{0}}\in V$. For the same reason we have $V^{k}\subseteq V^{k+1}$. Therefore there exists some $k\in\mathbb{N}$ for which the series converges and $V^{k+1}=V^{k}\equiv V$. Next we show that indeed there is no smaller subspace that is invariant under $\mathfrak{L}$ and contains $\ket{v_{0}}$. By definition of $V^{k}$, there exists a series of $A_{1},\ldots,A_{k}\in\mathfrak{L}$ such that any $\ket{v}\in V^{k}$ is connected to $\ket{v_{0}}$ via $A_{k}\ldots A_{1}\ket{v_{0}}=\ket{v_{1}}$. Therefore $V$ does not contain any subspace that would be invariant under $\mathfrak{L}$. Now assume the existence of smaller subspace $W$ invariant under $\mathfrak{L}$ that contains $\ket{v_{0}}$. Since we are looking for the smallest such subspace, $W$ cannot contain another subspace invariant under $\mathfrak{L}$. This implies that all $\ket{w}\in W$ are connected to $\ket{v_{0}}$ by some $A_{q}\ldots A_{1}\ket{w}=\ket{v_{0}}$. By assumption $V^{k}$ is converged, so that any such state $\ket{w}$ is also in $V$. Therefore $W\subseteq V$ and because neither contains additional invariant subspaces we have $W=V$.
\end{proof}
To set the stage for constructing $V$, we introduce the shorthand notation $K=\sum_{i=1}^{N}K^{i}$ with $K^{i}\equiv K_{\ket{\psi_{\text{AKLT}}}}^{\left(i,i+1\right)}$ and note that the local Hilbert space on a single link $\left(i,i+1\right)$ is partitioned into a local excited space and a local ground space, 
\begin{align*}
\mathcal{H}^{i}= & \mathcal{H}_{2}^{i}\oplus\mathcal{H}_{0,1}^{i}\,,
\end{align*}
where $\mathcal{H}_{2/\left(1,0\right)}^{i}\equiv\mathcal{H}_{2/\left(0,1\right)}^{\left(i,i+1\right)}$ refers to the $J=2$, respectively $J=0,1$ subspace.
The state $\ket{v_{0}}$ is excited with respect to this partition on all links of the chain and can be expressed locally as $\ket{J=2,m_{j}=2}$ on any link. To carry out the proof, we will first show controllability for even chain lengths and then use the result to prove controllability for odd chain lengths.

For even chain lengths, we can cover the whole chain by the set of adjacent non-overlapping even (odd) links,  denoting a link as even (odd) when the index of the left site is even. In the following, we focus on a cover of the chain by the set of odd links $\left(1,2\right),\left(3,4\right),\ldots,\left(N-1,N\right)$. The spin operators $J_{i}$ for even/ odd $i$ are then mutually commuting such that any state can be labeled by its $J,m_{j}$ quantum numbers, and we can partition the total Hilbert space, 
\begin{align*}
\mathcal{H} &=\bigotimes_{i=1,\text{ odd}}^{N}\left(\mathcal{H}_{2}^{i}\oplus\mathcal{H}_{0,1}^{i}\right)\\
 &=\sum_{k=0}^{N}\begin{pmatrix}N\\k\end{pmatrix}\underset{N-k}{\underbrace{\mathcal{H}_{2}\otimes\ldots\otimes\mathcal{H}_{2}}}\otimes\underset{k}{\underbrace{\mathcal{H}_{0,1}\otimes\ldots\otimes\mathcal{H}_{0,1}}}\\
 & =\sum_{k=0}^{N}\sum_{\left(p,q\right)\in P_{k}^{N}}\underset{N-k}{\underbrace{\mathcal{H}_{2}^{p\left(1\right)}\otimes\ldots\otimes\mathcal{H}_{2}^{p\left(N-k\right)}}}\\
 & \phantom{\sum_{k=0}^{N}\sum_{\left(p,q\right)\in P_{k}^{N}}}\otimes\underset{k}{\underbrace{\mathcal{H}_{0,1}^{q\left(1\right)}\otimes\ldots\otimes\mathcal{H}_{0,1}^{q\left(k\right)}}}\,,
\end{align*}
where $P_{k}^{N}$ denotes the set of permutations of an $N$-dimensional vector of $k$ entries equal to $0$ and $N-k$ entries equal to $1$. Here $q$, resp. $p$ in $\left(p,q\right)\in P_{k}^{N}$ denote the positions of the entries $1$, resp. $0$ in the permutation. We illustrate this for the example of an $N=4$ chain,
\begin{eqnarray*}
\mathcal{H}&= & \left(\mathcal{H}_{2}^{1}\oplus\mathcal{H}_{0,1}^{1}\right)\otimes\left(\mathcal{H}_{2}^{2}\oplus\mathcal{H}_{0,1}^{2}\right)\\
  &= & 
      \left(\mathcal{H}_{2}^{1}\otimes\mathcal{H}_{2}^{2}\right) \oplus\left(\mathcal{H}_{2}^{1}\otimes\mathcal{H}_{0,1}^{2}\right)\oplus\left(\mathcal{H}_{0,1}^{1}\otimes\mathcal{H}_{2}^{2}\right)\\
 & &\oplus\left(\mathcal{H}_{0,1}^{1}\otimes\mathcal{H}_{0,1}^{2}\right)\,.
\end{eqnarray*}
We will show that each such Hilbert space summand except $\mathcal{H}_{0}^{1}\otimes\ldots\otimes\mathcal{H}_{0}^{N-1}$ is connected to $\ket{v_{0}}$ by repeated actions $K$.

As a first step, consider the action of $K^{i}$ on $\ket{v_{0}}$. Note that $\ket{v_{0}}\in\mathcal{H}_{2}^{\text{odd}}\equiv\bigotimes_{i=1,\text{ odd}}^{N}\mathcal{H}_{2}^{i}$ and  $K^{i}$ contains any (Hermitian) operator on $\mathcal{H}_{2}^{i}$. Therefore the vector space orbit $K^{i}\ket{v_{0}}=\text{ span}\left\{ A\ket{v_{0}}:A\in K^{i}\right\} $ contains any state of the form $\ket{2,m_{j}}_{i}\bigotimes_{j=1j\neq i,j\text{ odd}}^{N}\ket{2,2}$ with $-2\leq m_{j}\leq2$, i.e., the whole $\mathcal{H}_{2}^{i}$. Since for every $i$ the vector space orbit $K^{i}\ket{v_{0}}$ is connected to $\ket{v_{0}}$ by $K$, it follows that they all belong to $V$ such that $\mathcal{H}_{2}^{\text{odd}}\subset V$.

Next, we show that we can generate a Hilbert space summand $\mathcal{H}_{0,1}^{i}$ at an arbitrary position $i$ from $\mathcal{H}_{2}^{\text{odd}}$. We consider the local kernelizer $K^{i+1}$ acting on the \textbf{even} link on $i+1$. Specifically we consider the image of $\mathcal{H}_{2}^{\text{odd}}$ under $K^{i+1}$. Since links $i$ and $i+1$, resp. $i$ and $i-1$, are overlapping, the action of $K^{i+1}$ is non-trivial only on $\mathcal{H}_{2}^{i}\otimes\mathcal{H}_{2}^{i+2}$ and it is sufficient to consider the vector space orbit $K^{i+1}\left[\mathcal{H}_{2}^{i}\otimes\mathcal{H}_{2}^{i+2}\right]$. One can show (numerically) that it has a finite overlap with $\mathcal{H}_{2}^{i}\otimes\mathcal{H}_{0,1}^{i+2}$ and  $\mathcal{H}_{0,1}^{i}\otimes\mathcal{H}_{2}^{i+2}$. More specifically, $K^{i+1}\left[\mathcal{H}_{2}^{i}\otimes\mathcal{H}_{2}^{i+2}\right]$ contains states of the form $\left\{ \underset{\in\mathcal{H}_{2}}{\ket{\phi_{j}}}\otimes\underset{\in\mathcal{H}_{0,1}}{\ket{\psi_{j}^{J=0,1}}}\right\} _{j=1,\ldots,4}$ where the four states  $\ket{\psi_{j}^{J=0,1}}$ form a basis of $\mathcal{H}_{0,1}$ and $\ket{\phi_{j}}$ are some states in $\mathcal{H}_{2}$. Similarly, it also contains $\left\{ \underset{\in\mathcal{H}_{0,1}}{\ket{\psi_{j}^{J=0,1}}}\otimes\underset{\in\mathcal{H}_{2}}{\ket{\phi_{j}}}\right\} _{j=1,\ldots,4}$. By another application of $K^{i}$, resp. $K^{i+2}$, acting only on $\ket{\phi_{j}}\in\mathcal{H}_{2}^{i}$, resp. $\mathcal{H}_{2}^{i+2}$, we generate both, $\mathcal{H}_{2}^{i}\otimes\mathcal{H}_{0}^{i+2}$ and $\mathcal{H}_{0}^{i}\otimes\mathcal{H}_{2}^{i+2}$. This works because $K^{i}$ contains any (Hermitian) operator on $\mathcal{H}_{2}$ such that $K^{i}\left[\ket{\phi_{j}}\otimes\ket{\psi_{j}^{J=0,1}}\right]=\mathcal{H}_{2}^{i}\otimes\ket{\psi_{j}^{J=0,1}}$ and likewise $K^{i+2}\left[\ket{\psi_{j}^{J=0,1}}\otimes\ket{\phi_{j}}\right]=\ket{\psi_{j}^{J=0,1}}\otimes\mathcal{H}_{2}^{i}$. 

This mechanism works in general and allows us to generate any Hilbert space summands with arbitrary number of $\mathcal{H}_{0,1}$ on arbitrary links by alternating application of the kernelizers on even and odd links. The only Hilbert space summand not entirely accessible in this way is $\mathcal{H}_{0,1}^{\text{odd}}\equiv\bigotimes_{i=1,\text{ odd}}^{N}\mathcal{H}_{0,1}^{i}$ and we conclude that $\mathcal{H}\setminus\mathcal{H}_{0,1}^{\text{odd}}\subset V$.

In order to connect all excited states in $\mathcal{H}_{0,1}^{\text{odd}}$ to $\ket{v_{0}}$, we now cover the whole chain by adjacent \textbf{even} links $\left(2,3\right),\left(3,4\right),\ldots,\left(N,1\right)$. First note that $\ket{v_{0}}\in\mathcal{H}_{2}^{\text{even}}$. Thus we can repeat the above procedure and immediately find that $\mathcal{H}\setminus\mathcal{H}_{0,1}^{\text{even}}\subset V$. In order to show that $V$ contains all states except $\ket{\psi_{\text{AKLT}}}$, we consider the orthogonal complement $V^{\perp}$. Note that $\left(\mathcal{H}\setminus\mathcal{H}_{0,1}^{\text{even}}\right)+\left(\mathcal{H}\setminus\mathcal{H}_{0,1}^{\text{even}}\right)\subseteq V$, therefore $V^{\perp}\subseteq\left(\left(\mathcal{H}\setminus\mathcal{H}_{0,1}^{\text{even}}\right)+\left(\mathcal{H}\setminus\mathcal{H}_{0,1}^{\text{even}}\right)\right)^{\perp}=\left(\mathcal{H}\setminus\mathcal{H}_{0,1}^{\text{even}}\right)^{\perp}\cap\left(\mathcal{H}\setminus\mathcal{H}_{0,1}^{\text{even}}\right)^{\perp}=\mathcal{H}_{0,1}^{\text{odd}}\cap\mathcal{H}_{0,1}^{\text{even}}$. The only state not in $V$ is simultaneously in the ground space of all even and all odd links. The only state in this set is the unique ground state $\ket{\psi_{\text{AKLT}}}$. This concludes the proof of controllability for even chain lengths.

For odd chain lengths, covering the the whole chain  by adjacent even/odd links leaves a single site uncovered, e.g. for $5$ sites one would cover $\left(1,2\right),\left(3,4\right)$ or $\left(2,3\right)\left(4,5\right)$. In other words, the chain essentially becomes a chain with open boundary conditions. We can still apply the procedure above but without the closing link $\left(N,1\right)$ there are now four states in $V^{\perp}$: the four degenerate AKLT ground states of a chain with open boundary conditions. To formalize this $V^{\perp}\subseteq\mathcal{H}_{0,1}^{\text{odd}}\cap\mathcal{H}_{0,1}^{\text{even}}$ is not cut with$\mathcal{H}_{0,1}^{N}$ and therefore the states in $V^{\perp}$ are not constraint to $\mathcal{H}_{0,1}^{N}$ resulting in the four degenerate AKLT states. We can partition the four states in $V^{\perp}$ into $3+1$ states: one state without excitations on link $N$ corresponding to the unique AKLT ground state of a chain with periodic boundary conditions. The other three are excited on link $N$ and are therefore subject to $K^{N}$ which allows us to connect them to $V$. This concludes the proof of controllability for odd chain lengths.

\nocite{*}

\bibliography{manuscript}

\begin{thebibliography}{76}%
\makeatletter
\providecommand \@ifxundefined [1]{%
 \@ifx{#1\undefined}
}%
\providecommand \@ifnum [1]{%
 \ifnum #1\expandafter \@firstoftwo
 \else \expandafter \@secondoftwo
 \fi
}%
\providecommand \@ifx [1]{%
 \ifx #1\expandafter \@firstoftwo
 \else \expandafter \@secondoftwo
 \fi
}%
\providecommand \natexlab [1]{#1}%
\providecommand \enquote  [1]{``#1''}%
\providecommand \bibnamefont  [1]{#1}%
\providecommand \bibfnamefont [1]{#1}%
\providecommand \citenamefont [1]{#1}%
\providecommand \href@noop [0]{\@secondoftwo}%
\providecommand \href [0]{\begingroup \@sanitize@url \@href}%
\providecommand \@href[1]{\@@startlink{#1}\@@href}%
\providecommand \@@href[1]{\endgroup#1\@@endlink}%
\providecommand \@sanitize@url [0]{\catcode `\\12\catcode `\$12\catcode
  `\&12\catcode `\#12\catcode `\^12\catcode `\_12\catcode `\%12\relax}%
\providecommand \@@startlink[1]{}%
\providecommand \@@endlink[0]{}%
\providecommand \url  [0]{\begingroup\@sanitize@url \@url }%
\providecommand \@url [1]{\endgroup\@href {#1}{\urlprefix }}%
\providecommand \urlprefix  [0]{URL }%
\providecommand \Eprint [0]{\href }%
\providecommand \doibase [0]{https://doi.org/}%
\providecommand \selectlanguage [0]{\@gobble}%
\providecommand \bibinfo  [0]{\@secondoftwo}%
\providecommand \bibfield  [0]{\@secondoftwo}%
\providecommand \translation [1]{[#1]}%
\providecommand \BibitemOpen [0]{}%
\providecommand \bibitemStop [0]{}%
\providecommand \bibitemNoStop [0]{.\EOS\space}%
\providecommand \EOS [0]{\spacefactor3000\relax}%
\providecommand \BibitemShut  [1]{\csname bibitem#1\endcsname}%
\let\auto@bib@innerbib\@empty
\bibitem [{\citenamefont {Poyatos}\ \emph {et~al.}(1996)\citenamefont
  {Poyatos}, \citenamefont {Cirac},\ and\ \citenamefont
  {Zoller}}]{PoyatosPRL96}%
  \BibitemOpen
  \bibfield  {author} {\bibinfo {author} {\bibfnamefont {J.~F.}\ \bibnamefont
  {Poyatos}}, \bibinfo {author} {\bibfnamefont {J.~I.}\ \bibnamefont {Cirac}},\
  and\ \bibinfo {author} {\bibfnamefont {P.}~\bibnamefont {Zoller}},\
  }\bibfield  {title} {\bibinfo {title} {Quantum reservoir engineering with
  laser cooled trapped ions},\ }\href
  {https://doi.org/10.1103/PhysRevLett.77.4728} {\bibfield  {journal} {\bibinfo
   {journal} {Phys. Rev. Lett.}\ }\textbf {\bibinfo {volume} {77}},\ \bibinfo
  {pages} {4728} (\bibinfo {year} {1996})}\BibitemShut {NoStop}%
\bibitem [{\citenamefont {Diehl}\ \emph {et~al.}(2008)\citenamefont {Diehl},
  \citenamefont {Micheli}, \citenamefont {Kantian}, \citenamefont {Kraus},
  \citenamefont {B\"uchler},\ and\ \citenamefont {Zoller}}]{DiehlNatPhys08}%
  \BibitemOpen
  \bibfield  {author} {\bibinfo {author} {\bibfnamefont {S.}~\bibnamefont
  {Diehl}}, \bibinfo {author} {\bibfnamefont {A.}~\bibnamefont {Micheli}},
  \bibinfo {author} {\bibfnamefont {A.}~\bibnamefont {Kantian}}, \bibinfo
  {author} {\bibfnamefont {B.}~\bibnamefont {Kraus}}, \bibinfo {author}
  {\bibfnamefont {H.-P.}\ \bibnamefont {B\"uchler}},\ and\ \bibinfo {author}
  {\bibfnamefont {P.}~\bibnamefont {Zoller}},\ }\bibfield  {title} {\bibinfo
  {title} {Quantum states and phases in driven open quantum systems with cold
  atoms},\ }\href {https://doi.org/10.1038/nphys1073} {\bibfield  {journal}
  {\bibinfo  {journal} {Nature Phys.}\ }\textbf {\bibinfo {volume} {4}},\
  \bibinfo {pages} {878} (\bibinfo {year} {2008})}\BibitemShut {NoStop}%
\bibitem [{\citenamefont {Kraus}\ \emph {et~al.}(2008)\citenamefont {Kraus},
  \citenamefont {Büchler}, \citenamefont {Diehl}, \citenamefont {Kantian},
  \citenamefont {Micheli},\ and\ \citenamefont
  {Zoller}}]{kraus_preparation_2008}%
  \BibitemOpen
  \bibfield  {author} {\bibinfo {author} {\bibfnamefont {B.}~\bibnamefont
  {Kraus}}, \bibinfo {author} {\bibfnamefont {H.~P.}\ \bibnamefont {Büchler}},
  \bibinfo {author} {\bibfnamefont {S.}~\bibnamefont {Diehl}}, \bibinfo
  {author} {\bibfnamefont {A.}~\bibnamefont {Kantian}}, \bibinfo {author}
  {\bibfnamefont {A.}~\bibnamefont {Micheli}},\ and\ \bibinfo {author}
  {\bibfnamefont {P.}~\bibnamefont {Zoller}},\ }\bibfield  {title} {\bibinfo
  {title} {Preparation of entangled states by quantum {M}arkov processes},\
  }\href {https://doi.org/10.1103/PhysRevA.78.042307} {\bibfield  {journal}
  {\bibinfo  {journal} {Physical Review A}\ }\textbf {\bibinfo {volume} {78}},\
  \bibinfo {pages} {042307} (\bibinfo {year} {2008})}\BibitemShut {NoStop}%
\bibitem [{\citenamefont {Verstraete}\ \emph {et~al.}(2009)\citenamefont
  {Verstraete}, \citenamefont {Wolf},\ and\ \citenamefont
  {Ignacio~Cirac}}]{verstraete_quantum_2009}%
  \BibitemOpen
  \bibfield  {author} {\bibinfo {author} {\bibfnamefont {F.}~\bibnamefont
  {Verstraete}}, \bibinfo {author} {\bibfnamefont {M.~M.}\ \bibnamefont
  {Wolf}},\ and\ \bibinfo {author} {\bibfnamefont {J.}~\bibnamefont
  {Ignacio~Cirac}},\ }\bibfield  {title} {\bibinfo {title} {Quantum computation
  and quantum-state engineering driven by dissipation},\ }\href
  {https://doi.org/10.1038/nphys1342} {\bibfield  {journal} {\bibinfo
  {journal} {Nature Physics}\ }\textbf {\bibinfo {volume} {5}},\ \bibinfo
  {pages} {633} (\bibinfo {year} {2009})}\BibitemShut {NoStop}%
\bibitem [{\citenamefont {Diehl}\ \emph {et~al.}(2011)\citenamefont {Diehl},
  \citenamefont {Rico}, \citenamefont {Baranov},\ and\ \citenamefont
  {Zoller}}]{DiehlNatPhys2011}%
  \BibitemOpen
  \bibfield  {author} {\bibinfo {author} {\bibfnamefont {S.}~\bibnamefont
  {Diehl}}, \bibinfo {author} {\bibfnamefont {E.}~\bibnamefont {Rico}},
  \bibinfo {author} {\bibfnamefont {M.~A.}\ \bibnamefont {Baranov}},\ and\
  \bibinfo {author} {\bibfnamefont {P.}~\bibnamefont {Zoller}},\ }\bibfield
  {title} {\bibinfo {title} {Topology by dissipation in atomic quantum wires},\
  }\href {https://doi.org/10.1038/nphys2106} {\bibfield  {journal} {\bibinfo
  {journal} {Nature Phys.}\ }\textbf {\bibinfo {volume} {7}},\ \bibinfo {pages}
  {971} (\bibinfo {year} {2011})}\BibitemShut {NoStop}%
\bibitem [{\citenamefont {{Mi \textit{et al.}, Google Quantum AI and
  Collaborators}}(2023)}]{mi2023}%
  \BibitemOpen
  \bibfield  {author} {\bibinfo {author} {\bibfnamefont {X.}~\bibnamefont {{Mi
  \textit{et al.}, Google Quantum AI and Collaborators}}},\ }\href
  {https://doi.org/10.48550/arXiv.2304.13878} {\bibinfo {title} {Stable
  {{Quantum-Correlated Many Body States}} via {{Engineered Dissipation}}}}
  (\bibinfo {year} {2023}),\ \Eprint {https://arxiv.org/abs/2304.13878}
  {2304.13878 [quant-ph]} \BibitemShut {NoStop}%
\bibitem [{\citenamefont {Brown}\ \emph {et~al.}(2022)\citenamefont {Brown},
  \citenamefont {Doucet}, \citenamefont {Rist{\`e}}, \citenamefont {Ribeill},
  \citenamefont {Cicak}, \citenamefont {Aumentado}, \citenamefont {Simmonds},
  \citenamefont {Govia}, \citenamefont {Kamal},\ and\ \citenamefont
  {Ranzani}}]{Brown_etal-Kamal:2022}%
  \BibitemOpen
  \bibfield  {author} {\bibinfo {author} {\bibfnamefont {T.}~\bibnamefont
  {Brown}}, \bibinfo {author} {\bibfnamefont {E.}~\bibnamefont {Doucet}},
  \bibinfo {author} {\bibfnamefont {D.}~\bibnamefont {Rist{\`e}}}, \bibinfo
  {author} {\bibfnamefont {G.}~\bibnamefont {Ribeill}}, \bibinfo {author}
  {\bibfnamefont {K.}~\bibnamefont {Cicak}}, \bibinfo {author} {\bibfnamefont
  {J.}~\bibnamefont {Aumentado}}, \bibinfo {author} {\bibfnamefont
  {R.}~\bibnamefont {Simmonds}}, \bibinfo {author} {\bibfnamefont
  {L.}~\bibnamefont {Govia}}, \bibinfo {author} {\bibfnamefont
  {A.}~\bibnamefont {Kamal}},\ and\ \bibinfo {author} {\bibfnamefont
  {L.}~\bibnamefont {Ranzani}},\ }\bibfield  {title} {\bibinfo {title} {Trade
  off-free entanglement stabilization in a superconducting qutrit-qubit
  system},\ }\href {https://doi.org/10.1038/s41467-022-31638-0} {\bibfield
  {journal} {\bibinfo  {journal} {Nature Communications}\ }\textbf {\bibinfo
  {volume} {13}},\ \bibinfo {pages} {3994} (\bibinfo {year}
  {2022})}\BibitemShut {NoStop}%
\bibitem [{\citenamefont {Harrington}\ \emph {et~al.}(2022)\citenamefont
  {Harrington}, \citenamefont {Mueller},\ and\ \citenamefont
  {Murch}}]{HarringtonNatRevPhys2022}%
  \BibitemOpen
  \bibfield  {author} {\bibinfo {author} {\bibfnamefont {P.~M.}\ \bibnamefont
  {Harrington}}, \bibinfo {author} {\bibfnamefont {E.~J.}\ \bibnamefont
  {Mueller}},\ and\ \bibinfo {author} {\bibfnamefont {K.~W.}\ \bibnamefont
  {Murch}},\ }\bibfield  {title} {\bibinfo {title} {Engineered dissipation for
  quantum information science},\ }\href
  {https://doi.org/10.1038/s42254-022-00494-8} {\bibfield  {journal} {\bibinfo
  {journal} {Nature Rev. Phys.}\ }\textbf {\bibinfo {volume} {4}},\ \bibinfo
  {pages} {660} (\bibinfo {year} {2022})}\BibitemShut {NoStop}%
\bibitem [{\citenamefont {Wiseman}\ and\ \citenamefont
  {Milburn}(2009)}]{wisemanBook}%
  \BibitemOpen
  \bibfield  {author} {\bibinfo {author} {\bibfnamefont {H.~M.}\ \bibnamefont
  {Wiseman}}\ and\ \bibinfo {author} {\bibfnamefont {G.~J.}\ \bibnamefont
  {Milburn}},\ }\href {https://doi.org/10.1017/CBO9780511813948} {\emph
  {\bibinfo {title} {Quantum Measurement and Control}}}\ (\bibinfo  {publisher}
  {Cambridge University Press},\ \bibinfo {year} {2009})\BibitemShut {NoStop}%
\bibitem [{\citenamefont {Haroche}\ and\ \citenamefont
  {Raimond}(2006)}]{HarocheBook}%
  \BibitemOpen
  \bibfield  {author} {\bibinfo {author} {\bibfnamefont {S.}~\bibnamefont
  {Haroche}}\ and\ \bibinfo {author} {\bibfnamefont {J.~M.}\ \bibnamefont
  {Raimond}},\ }\href@noop {} {\emph {\bibinfo {title} {Exploring the Quantum:
  {A}toms, Cavities, and Photons}}}\ (\bibinfo  {publisher} {Oxford Univ.
  Press},\ \bibinfo {address} {Oxford},\ \bibinfo {year} {2006})\BibitemShut
  {NoStop}%
\bibitem [{\citenamefont {Pielawa}\ \emph {et~al.}(2007)\citenamefont
  {Pielawa}, \citenamefont {Morigi}, \citenamefont {Vitali},\ and\
  \citenamefont {Davidovich}}]{Pielawa:2007}%
  \BibitemOpen
  \bibfield  {author} {\bibinfo {author} {\bibfnamefont {S.}~\bibnamefont
  {Pielawa}}, \bibinfo {author} {\bibfnamefont {G.}~\bibnamefont {Morigi}},
  \bibinfo {author} {\bibfnamefont {D.}~\bibnamefont {Vitali}},\ and\ \bibinfo
  {author} {\bibfnamefont {L.}~\bibnamefont {Davidovich}},\ }\bibfield  {title}
  {\bibinfo {title} {{Generation of Einstein-Podolsky-Rosen-Entangled Radiation
  through an Atomic Reservoir}},\ }\href
  {https://doi.org/10.1103/PhysRevLett.98.240401} {\bibfield  {journal}
  {\bibinfo  {journal} {Phys. Rev. Lett.}\ }\textbf {\bibinfo {volume} {98}},\
  \bibinfo {pages} {240401} (\bibinfo {year} {2007})}\BibitemShut {NoStop}%
\bibitem [{\citenamefont {Raghunandan}\ \emph {et~al.}(2020)\citenamefont
  {Raghunandan}, \citenamefont {Wolf}, \citenamefont {Ospelkaus}, \citenamefont
  {Schmidt},\ and\ \citenamefont {Weimer}}]{raghunandan2020}%
  \BibitemOpen
  \bibfield  {author} {\bibinfo {author} {\bibfnamefont {M.}~\bibnamefont
  {Raghunandan}}, \bibinfo {author} {\bibfnamefont {F.}~\bibnamefont {Wolf}},
  \bibinfo {author} {\bibfnamefont {C.}~\bibnamefont {Ospelkaus}}, \bibinfo
  {author} {\bibfnamefont {P.~O.}\ \bibnamefont {Schmidt}},\ and\ \bibinfo
  {author} {\bibfnamefont {H.}~\bibnamefont {Weimer}},\ }\bibfield  {title}
  {\bibinfo {title} {Initialization of quantum simulators by sympathetic
  cooling},\ }\href {https://doi.org/10.1126/sciadv.aaw9268} {\bibfield
  {journal} {\bibinfo  {journal} {Science Advances}\ }\textbf {\bibinfo
  {volume} {6}},\ \bibinfo {pages} {eaaw9268} (\bibinfo {year}
  {2020})}\BibitemShut {NoStop}%
\bibitem [{\citenamefont {Burgarth}\ and\ \citenamefont
  {Giovannetti}(2007{\natexlab{a}})}]{BurgarthPRL2007}%
  \BibitemOpen
  \bibfield  {author} {\bibinfo {author} {\bibfnamefont {D.}~\bibnamefont
  {Burgarth}}\ and\ \bibinfo {author} {\bibfnamefont {V.}~\bibnamefont
  {Giovannetti}},\ }\bibfield  {title} {\bibinfo {title} {Full control by
  locally induced relaxation},\ }\href
  {https://doi.org/10.1103/PhysRevLett.99.100501} {\bibfield  {journal}
  {\bibinfo  {journal} {Phys. Rev. Lett.}\ }\textbf {\bibinfo {volume} {99}},\
  \bibinfo {pages} {100501} (\bibinfo {year} {2007}{\natexlab{a}})}\BibitemShut
  {NoStop}%
\bibitem [{\citenamefont {Burgarth}\ and\ \citenamefont
  {Giovannetti}(2007{\natexlab{b}})}]{BurgarthPRA2007}%
  \BibitemOpen
  \bibfield  {author} {\bibinfo {author} {\bibfnamefont {D.}~\bibnamefont
  {Burgarth}}\ and\ \bibinfo {author} {\bibfnamefont {V.}~\bibnamefont
  {Giovannetti}},\ }\bibfield  {title} {\bibinfo {title} {Mediated
  homogenization},\ }\href {https://doi.org/10.1103/PhysRevA.76.062307}
  {\bibfield  {journal} {\bibinfo  {journal} {Phys. Rev. A}\ }\textbf {\bibinfo
  {volume} {76}},\ \bibinfo {pages} {062307} (\bibinfo {year}
  {2007}{\natexlab{b}})}\BibitemShut {NoStop}%
\bibitem [{\citenamefont {Boykin}\ \emph {et~al.}(2002)\citenamefont {Boykin},
  \citenamefont {Mor}, \citenamefont {Roychowdhury}, \citenamefont {Vatan},\
  and\ \citenamefont {Vrijen}}]{Boykin2002}%
  \BibitemOpen
  \bibfield  {author} {\bibinfo {author} {\bibfnamefont {P.~O.}\ \bibnamefont
  {Boykin}}, \bibinfo {author} {\bibfnamefont {T.}~\bibnamefont {Mor}},
  \bibinfo {author} {\bibfnamefont {V.}~\bibnamefont {Roychowdhury}}, \bibinfo
  {author} {\bibfnamefont {F.}~\bibnamefont {Vatan}},\ and\ \bibinfo {author}
  {\bibfnamefont {R.}~\bibnamefont {Vrijen}},\ }\bibfield  {title} {\bibinfo
  {title} {Algorithmic cooling and scalable {NMR} quantum computers},\ }\href
  {https://doi.org/10.1073/pnas.241641898} {\bibfield  {journal} {\bibinfo
  {journal} {Proceedings of the National Academy of Sciences}\ }\textbf
  {\bibinfo {volume} {99}},\ \bibinfo {pages} {3388} (\bibinfo {year}
  {2002})}\BibitemShut {NoStop}%
\bibitem [{\citenamefont {Metcalf}\ \emph {et~al.}(2020)\citenamefont
  {Metcalf}, \citenamefont {Moussa}, \citenamefont {de~Jong},\ and\
  \citenamefont {Sarovar}}]{Metcalf2020}%
  \BibitemOpen
  \bibfield  {author} {\bibinfo {author} {\bibfnamefont {M.}~\bibnamefont
  {Metcalf}}, \bibinfo {author} {\bibfnamefont {J.~E.}\ \bibnamefont {Moussa}},
  \bibinfo {author} {\bibfnamefont {W.~A.}\ \bibnamefont {de~Jong}},\ and\
  \bibinfo {author} {\bibfnamefont {M.}~\bibnamefont {Sarovar}},\ }\bibfield
  {title} {\bibinfo {title} {Engineered thermalization and cooling of quantum
  many-body systems},\ }\href
  {https://doi.org/10.1103/PhysRevResearch.2.023214} {\bibfield  {journal}
  {\bibinfo  {journal} {Phys. Rev. Res.}\ }\textbf {\bibinfo {volume} {2}},\
  \bibinfo {pages} {023214} (\bibinfo {year} {2020})}\BibitemShut {NoStop}%
\bibitem [{\citenamefont {Zaletel}\ \emph {et~al.}(2021)\citenamefont
  {Zaletel}, \citenamefont {Kaufman}, \citenamefont {Stamper-Kurn},\ and\
  \citenamefont {Yao}}]{Zaletel2021}%
  \BibitemOpen
  \bibfield  {author} {\bibinfo {author} {\bibfnamefont {M.~P.}\ \bibnamefont
  {Zaletel}}, \bibinfo {author} {\bibfnamefont {A.}~\bibnamefont {Kaufman}},
  \bibinfo {author} {\bibfnamefont {D.~M.}\ \bibnamefont {Stamper-Kurn}},\ and\
  \bibinfo {author} {\bibfnamefont {N.~Y.}\ \bibnamefont {Yao}},\ }\bibfield
  {title} {\bibinfo {title} {Preparation of low entropy correlated many-body
  states via conformal cooling quenches},\ }\href
  {https://doi.org/10.1103/PhysRevLett.126.103401} {\bibfield  {journal}
  {\bibinfo  {journal} {Phys. Rev. Lett.}\ }\textbf {\bibinfo {volume} {126}},\
  \bibinfo {pages} {103401} (\bibinfo {year} {2021})}\BibitemShut {NoStop}%
\bibitem [{\citenamefont {Rojan}\ \emph {et~al.}(2014)\citenamefont {Rojan},
  \citenamefont {Reich}, \citenamefont {Dotsenko}, \citenamefont {Raimond},
  \citenamefont {Koch},\ and\ \citenamefont {Morigi}}]{RojanPRA2014}%
  \BibitemOpen
  \bibfield  {author} {\bibinfo {author} {\bibfnamefont {K.}~\bibnamefont
  {Rojan}}, \bibinfo {author} {\bibfnamefont {D.~M.}\ \bibnamefont {Reich}},
  \bibinfo {author} {\bibfnamefont {I.}~\bibnamefont {Dotsenko}}, \bibinfo
  {author} {\bibfnamefont {J.-M.}\ \bibnamefont {Raimond}}, \bibinfo {author}
  {\bibfnamefont {C.~P.}\ \bibnamefont {Koch}},\ and\ \bibinfo {author}
  {\bibfnamefont {G.}~\bibnamefont {Morigi}},\ }\bibfield  {title} {\bibinfo
  {title} {Arbitrary-quantum-state preparation of a harmonic oscillator via
  optimal control},\ }\href {https://doi.org/10.1103/PhysRevA.90.023824}
  {\bibfield  {journal} {\bibinfo  {journal} {Phys. Rev. A}\ }\textbf {\bibinfo
  {volume} {90}},\ \bibinfo {pages} {023824} (\bibinfo {year}
  {2014})}\BibitemShut {NoStop}%
\bibitem [{\citenamefont {Roy}\ \emph {et~al.}(2020)\citenamefont {Roy},
  \citenamefont {Chalker}, \citenamefont {Gornyi},\ and\ \citenamefont
  {Gefen}}]{roy_measurement-induced_2020}%
  \BibitemOpen
  \bibfield  {author} {\bibinfo {author} {\bibfnamefont {S.}~\bibnamefont
  {Roy}}, \bibinfo {author} {\bibfnamefont {J.~T.}\ \bibnamefont {Chalker}},
  \bibinfo {author} {\bibfnamefont {I.~V.}\ \bibnamefont {Gornyi}},\ and\
  \bibinfo {author} {\bibfnamefont {Y.}~\bibnamefont {Gefen}},\ }\bibfield
  {title} {\bibinfo {title} {Measurement-induced steering of quantum systems},\
  }\href {https://doi.org/10.1103/PhysRevResearch.2.033347} {\bibfield
  {journal} {\bibinfo  {journal} {Physical Review Research}\ }\textbf {\bibinfo
  {volume} {2}},\ \bibinfo {pages} {033347} (\bibinfo {year}
  {2020})}\BibitemShut {NoStop}%
\bibitem [{\citenamefont {Kumar}\ \emph {et~al.}(2020)\citenamefont {Kumar},
  \citenamefont {Snizhko},\ and\ \citenamefont {Gefen}}]{KumarPRR2020}%
  \BibitemOpen
  \bibfield  {author} {\bibinfo {author} {\bibfnamefont {P.}~\bibnamefont
  {Kumar}}, \bibinfo {author} {\bibfnamefont {K.}~\bibnamefont {Snizhko}},\
  and\ \bibinfo {author} {\bibfnamefont {Y.}~\bibnamefont {Gefen}},\ }\bibfield
   {title} {\bibinfo {title} {Engineering two-qubit mixed states with weak
  measurements},\ }\href {https://doi.org/10.1103/PhysRevResearch.2.042014}
  {\bibfield  {journal} {\bibinfo  {journal} {Phys. Rev. Res.}\ }\textbf
  {\bibinfo {volume} {2}},\ \bibinfo {pages} {042014} (\bibinfo {year}
  {2020})}\BibitemShut {NoStop}%
\bibitem [{\citenamefont {Perez-Garcia}\ \emph {et~al.}(2008)\citenamefont
  {Perez-Garcia}, \citenamefont {Verstraete}, \citenamefont {Wolf},\ and\
  \citenamefont {Cirac}}]{Perez-GarciaQIC2008}%
  \BibitemOpen
  \bibfield  {author} {\bibinfo {author} {\bibfnamefont {D.}~\bibnamefont
  {Perez-Garcia}}, \bibinfo {author} {\bibfnamefont {F.}~\bibnamefont
  {Verstraete}}, \bibinfo {author} {\bibfnamefont {M.~M.}\ \bibnamefont
  {Wolf}},\ and\ \bibinfo {author} {\bibfnamefont {J.~I.}\ \bibnamefont
  {Cirac}},\ }\bibfield  {title} {\bibinfo {title} {Peps as unique ground
  states of local {H}amiltonians},\ }\href
  {https://doi.org/10.26421/QIC8.6-7-6} {\bibfield  {journal} {\bibinfo
  {journal} {Quantum Info. Comput.}\ }\textbf {\bibinfo {volume} {8}},\
  \bibinfo {pages} {650–663} (\bibinfo {year} {2008})}\BibitemShut {NoStop}%
\bibitem [{\citenamefont {Johnson}\ \emph {et~al.}(2016)\citenamefont
  {Johnson}, \citenamefont {Ticozzi},\ and\ \citenamefont
  {Viola}}]{JohnsonQIP2016}%
  \BibitemOpen
  \bibfield  {author} {\bibinfo {author} {\bibfnamefont {P.~D.}\ \bibnamefont
  {Johnson}}, \bibinfo {author} {\bibfnamefont {F.}~\bibnamefont {Ticozzi}},\
  and\ \bibinfo {author} {\bibfnamefont {L.}~\bibnamefont {Viola}},\ }\bibfield
   {title} {\bibinfo {title} {General fixed points of quasi-local
  frustration-free quantum semigroups: From invariance to stabilization},\
  }\href {https://doi.org/10.26421/QIC16.7-8-5} {\bibfield  {journal} {\bibinfo
   {journal} {Quantum Info. Comput.}\ }\textbf {\bibinfo {volume} {16}},\
  \bibinfo {pages} {657–699} (\bibinfo {year} {2016})}\BibitemShut {NoStop}%
\bibitem [{\citenamefont {Ticozzi}\ and\ \citenamefont
  {Viola}(2014)}]{ticozzi_steady-state_2014}%
  \BibitemOpen
  \bibfield  {author} {\bibinfo {author} {\bibfnamefont {F.}~\bibnamefont
  {Ticozzi}}\ and\ \bibinfo {author} {\bibfnamefont {L.}~\bibnamefont
  {Viola}},\ }\bibfield  {title} {\bibinfo {title} {Steady-state entanglement
  by engineered quasi-local {M}arkovian dissipation: {H}amiltonian-assisted and
  conditional stabilization},\ }\href {https://doi.org/10.26421/QIC14.3-4-5}
  {\bibfield  {journal} {\bibinfo  {journal} {Quantum Info. Comput.}\ }\textbf
  {\bibinfo {volume} {14}},\ \bibinfo {pages} {265} (\bibinfo {year}
  {2014})}\BibitemShut {NoStop}%
\bibitem [{\citenamefont {Fisher}\ \emph {et~al.}(2023)\citenamefont {Fisher},
  \citenamefont {Khemani}, \citenamefont {Nahum},\ and\ \citenamefont
  {Vijay}}]{FisherAnnuRev2023}%
  \BibitemOpen
  \bibfield  {author} {\bibinfo {author} {\bibfnamefont {M.~P.}\ \bibnamefont
  {Fisher}}, \bibinfo {author} {\bibfnamefont {V.}~\bibnamefont {Khemani}},
  \bibinfo {author} {\bibfnamefont {A.}~\bibnamefont {Nahum}},\ and\ \bibinfo
  {author} {\bibfnamefont {S.}~\bibnamefont {Vijay}},\ }\bibfield  {title}
  {\bibinfo {title} {Random quantum circuits},\ }\href
  {https://doi.org/10.1146/annurev-conmatphys-031720-030658} {\bibfield
  {journal} {\bibinfo  {journal} {Annual Review of Condensed Matter Physics}\
  }\textbf {\bibinfo {volume} {14}},\ \bibinfo {pages} {335} (\bibinfo {year}
  {2023})}\BibitemShut {NoStop}%
\bibitem [{\citenamefont {Noel}\ \emph {et~al.}(2022)\citenamefont {Noel},
  \citenamefont {Niroula}, \citenamefont {Zhu}, \citenamefont {Risinger},
  \citenamefont {Egan}, \citenamefont {Biswas}, \citenamefont {Cetina},
  \citenamefont {Gorshkov}, \citenamefont {Gullans}, \citenamefont {Huse},\
  and\ \citenamefont {Monroe}}]{Noel2022a}%
  \BibitemOpen
  \bibfield  {author} {\bibinfo {author} {\bibfnamefont {C.}~\bibnamefont
  {Noel}}, \bibinfo {author} {\bibfnamefont {P.}~\bibnamefont {Niroula}},
  \bibinfo {author} {\bibfnamefont {D.}~\bibnamefont {Zhu}}, \bibinfo {author}
  {\bibfnamefont {A.}~\bibnamefont {Risinger}}, \bibinfo {author}
  {\bibfnamefont {L.}~\bibnamefont {Egan}}, \bibinfo {author} {\bibfnamefont
  {D.}~\bibnamefont {Biswas}}, \bibinfo {author} {\bibfnamefont
  {M.}~\bibnamefont {Cetina}}, \bibinfo {author} {\bibfnamefont {A.~V.}\
  \bibnamefont {Gorshkov}}, \bibinfo {author} {\bibfnamefont {M.~J.}\
  \bibnamefont {Gullans}}, \bibinfo {author} {\bibfnamefont {D.~A.}\
  \bibnamefont {Huse}},\ and\ \bibinfo {author} {\bibfnamefont
  {C.}~\bibnamefont {Monroe}},\ }\bibfield  {title} {\bibinfo {title}
  {Measurement-induced quantum phases realized in a trapped-ion quantum
  computer},\ }\href {https://doi.org/10.1038/s41567-022-01619-7} {\bibfield
  {journal} {\bibinfo  {journal} {Nat. Phys.}\ }\textbf {\bibinfo {volume}
  {18}},\ \bibinfo {pages} {760} (\bibinfo {year} {2022})}\BibitemShut
  {NoStop}%
\bibitem [{\citenamefont {Koh}\ \emph {et~al.}(2023)\citenamefont {Koh},
  \citenamefont {Sun}, \citenamefont {Motta},\ and\ \citenamefont
  {Minnich}}]{Koh2022}%
  \BibitemOpen
  \bibfield  {author} {\bibinfo {author} {\bibfnamefont {J.~M.}\ \bibnamefont
  {Koh}}, \bibinfo {author} {\bibfnamefont {S.-N.}\ \bibnamefont {Sun}},
  \bibinfo {author} {\bibfnamefont {M.}~\bibnamefont {Motta}},\ and\ \bibinfo
  {author} {\bibfnamefont {A.~J.}\ \bibnamefont {Minnich}},\ }\bibfield
  {title} {\bibinfo {title} {Measurement-induced entanglement phase transition
  on a superconducting quantum processor with mid-circuit readout},\ }\bibfield
   {journal} {\bibinfo  {journal} {Nat. Phys.}\ }\textbf {\bibinfo {volume}
  {19}},\ \href {https://doi.org/10.1038/s41567-023-02076-6}
  {10.1038/s41567-023-02076-6} (\bibinfo {year} {2023})\BibitemShut {NoStop}%
\bibitem [{\citenamefont {{Hoke \textit{et al.}, Google Quantum AI and
  Collaborators}}(2023)}]{hoke2023}%
  \BibitemOpen
  \bibfield  {author} {\bibinfo {author} {\bibfnamefont {J.~C.}\ \bibnamefont
  {{Hoke \textit{et al.}, Google Quantum AI and Collaborators}}},\ }\bibfield
  {title} {\bibinfo {title} {Measurement-induced entanglement and teleportation
  on a noisy quantum processor},\ }\href
  {https://doi.org/10.1038/s41586-023-06505-7} {\bibfield  {journal} {\bibinfo
  {journal} {Nature}\ }\textbf {\bibinfo {volume} {622}},\ \bibinfo {pages}
  {481} (\bibinfo {year} {2023})}\BibitemShut {NoStop}%
\bibitem [{\citenamefont {Polla}\ \emph {et~al.}(2021)\citenamefont {Polla},
  \citenamefont {Herasymenko},\ and\ \citenamefont {O'Brien}}]{PollaPRA2021}%
  \BibitemOpen
  \bibfield  {author} {\bibinfo {author} {\bibfnamefont {S.}~\bibnamefont
  {Polla}}, \bibinfo {author} {\bibfnamefont {Y.}~\bibnamefont {Herasymenko}},\
  and\ \bibinfo {author} {\bibfnamefont {T.~E.}\ \bibnamefont {O'Brien}},\
  }\bibfield  {title} {\bibinfo {title} {Quantum digital cooling},\ }\href
  {https://doi.org/10.1103/PhysRevA.104.012414} {\bibfield  {journal} {\bibinfo
   {journal} {Phys. Rev. A}\ }\textbf {\bibinfo {volume} {104}},\ \bibinfo
  {pages} {012414} (\bibinfo {year} {2021})}\BibitemShut {NoStop}%
\bibitem [{\citenamefont {Feng}\ \emph {et~al.}(2022)\citenamefont {Feng},
  \citenamefont {Wu},\ and\ \citenamefont {Wilczek}}]{FengPRA2022}%
  \BibitemOpen
  \bibfield  {author} {\bibinfo {author} {\bibfnamefont {J.-J.}\ \bibnamefont
  {Feng}}, \bibinfo {author} {\bibfnamefont {B.}~\bibnamefont {Wu}},\ and\
  \bibinfo {author} {\bibfnamefont {F.}~\bibnamefont {Wilczek}},\ }\bibfield
  {title} {\bibinfo {title} {Quantum computing by coherent cooling},\ }\href
  {https://doi.org/10.1103/PhysRevA.105.052601} {\bibfield  {journal} {\bibinfo
   {journal} {Phys. Rev. A}\ }\textbf {\bibinfo {volume} {105}},\ \bibinfo
  {pages} {052601} (\bibinfo {year} {2022})}\BibitemShut {NoStop}%
\bibitem [{\citenamefont {Matthies}\ \emph {et~al.}(2023)\citenamefont
  {Matthies}, \citenamefont {Rudner}, \citenamefont {Rosch},\ and\
  \citenamefont {Berg}}]{matthies2023}%
  \BibitemOpen
  \bibfield  {author} {\bibinfo {author} {\bibfnamefont {A.}~\bibnamefont
  {Matthies}}, \bibinfo {author} {\bibfnamefont {M.}~\bibnamefont {Rudner}},
  \bibinfo {author} {\bibfnamefont {A.}~\bibnamefont {Rosch}},\ and\ \bibinfo
  {author} {\bibfnamefont {E.}~\bibnamefont {Berg}},\ }\href@noop {} {\bibinfo
  {title} {Programmable adiabatic demagnetization for systems with trivial and
  topological excitations}} (\bibinfo {year} {2023}),\ \Eprint
  {https://arxiv.org/abs/2210.17256} {arXiv:2210.17256 [quant-ph]} \BibitemShut
  {NoStop}%
\bibitem [{\citenamefont {Kishony}\ \emph {et~al.}(2023)\citenamefont
  {Kishony}, \citenamefont {Rudner}, \citenamefont {Rosch},\ and\ \citenamefont
  {Berg}}]{kishony2023}%
  \BibitemOpen
  \bibfield  {author} {\bibinfo {author} {\bibfnamefont {G.}~\bibnamefont
  {Kishony}}, \bibinfo {author} {\bibfnamefont {M.~S.}\ \bibnamefont {Rudner}},
  \bibinfo {author} {\bibfnamefont {A.}~\bibnamefont {Rosch}},\ and\ \bibinfo
  {author} {\bibfnamefont {E.}~\bibnamefont {Berg}},\ }\href@noop {} {\bibinfo
  {title} {Gauged cooling of topological excitations and emergent fermions on
  quantum simulators}} (\bibinfo {year} {2023}),\ \Eprint
  {https://arxiv.org/abs/2310.16082} {arXiv:2310.16082 [cond-mat.str-el]}
  \BibitemShut {NoStop}%
\bibitem [{\citenamefont {Maurya}\ \emph {et~al.}(2023)\citenamefont {Maurya},
  \citenamefont {Zhang}, \citenamefont {Kowsari}, \citenamefont {Kuo},
  \citenamefont {Hartsell}, \citenamefont {Miyamoto}, \citenamefont {Liu},
  \citenamefont {Shanto}, \citenamefont {Zarassi}, \citenamefont {Murch},\ and\
  \citenamefont {Levenson-Falk}}]{maurya2023}%
  \BibitemOpen
  \bibfield  {author} {\bibinfo {author} {\bibfnamefont {V.}~\bibnamefont
  {Maurya}}, \bibinfo {author} {\bibfnamefont {H.}~\bibnamefont {Zhang}},
  \bibinfo {author} {\bibfnamefont {D.}~\bibnamefont {Kowsari}}, \bibinfo
  {author} {\bibfnamefont {A.}~\bibnamefont {Kuo}}, \bibinfo {author}
  {\bibfnamefont {D.~M.}\ \bibnamefont {Hartsell}}, \bibinfo {author}
  {\bibfnamefont {C.}~\bibnamefont {Miyamoto}}, \bibinfo {author}
  {\bibfnamefont {J.}~\bibnamefont {Liu}}, \bibinfo {author} {\bibfnamefont
  {S.}~\bibnamefont {Shanto}}, \bibinfo {author} {\bibfnamefont
  {A.}~\bibnamefont {Zarassi}}, \bibinfo {author} {\bibfnamefont {K.~W.}\
  \bibnamefont {Murch}},\ and\ \bibinfo {author} {\bibfnamefont {E.~M.}\
  \bibnamefont {Levenson-Falk}},\ }\href@noop {} {\bibinfo {title} {On-demand
  driven dissipation for cavity reset and cooling}} (\bibinfo {year} {2023}),\
  \Eprint {https://arxiv.org/abs/2310.16785} {arXiv:2310.16785 [quant-ph]}
  \BibitemShut {NoStop}%
\bibitem [{\citenamefont {Biella}\ and\ \citenamefont
  {Schir{\'{o}}}(2021)}]{Biella2021}%
  \BibitemOpen
  \bibfield  {author} {\bibinfo {author} {\bibfnamefont {A.}~\bibnamefont
  {Biella}}\ and\ \bibinfo {author} {\bibfnamefont {M.}~\bibnamefont
  {Schir{\'{o}}}},\ }\bibfield  {title} {\bibinfo {title} {Many-{B}ody
  {Q}uantum {Z}eno {E}ffect and {M}easurement-{I}nduced {S}ubradiance
  {T}ransition},\ }\href {https://doi.org/10.22331/q-2021-08-19-528} {\bibfield
   {journal} {\bibinfo  {journal} {{Quantum}}\ }\textbf {\bibinfo {volume}
  {5}},\ \bibinfo {pages} {528} (\bibinfo {year} {2021})}\BibitemShut {NoStop}%
\bibitem [{\citenamefont {Albash}\ and\ \citenamefont
  {Lidar}(2018)}]{Albash_Lidar18}%
  \BibitemOpen
  \bibfield  {author} {\bibinfo {author} {\bibfnamefont {T.}~\bibnamefont
  {Albash}}\ and\ \bibinfo {author} {\bibfnamefont {D.~A.}\ \bibnamefont
  {Lidar}},\ }\bibfield  {title} {\bibinfo {title} {Adiabatic quantum
  computation},\ }\href {https://doi.org/10.1103/RevModPhys.90.015002}
  {\bibfield  {journal} {\bibinfo  {journal} {Rev. Mod. Phys.}\ }\textbf
  {\bibinfo {volume} {90}},\ \bibinfo {pages} {015002} (\bibinfo {year}
  {2018})}\BibitemShut {NoStop}%
\bibitem [{\citenamefont {Menu}\ \emph {et~al.}(2022)\citenamefont {Menu},
  \citenamefont {Langbehn}, \citenamefont {Koch},\ and\ \citenamefont
  {Morigi}}]{Menu_etal22}%
  \BibitemOpen
  \bibfield  {author} {\bibinfo {author} {\bibfnamefont {R.}~\bibnamefont
  {Menu}}, \bibinfo {author} {\bibfnamefont {J.}~\bibnamefont {Langbehn}},
  \bibinfo {author} {\bibfnamefont {C.~P.}\ \bibnamefont {Koch}},\ and\
  \bibinfo {author} {\bibfnamefont {G.}~\bibnamefont {Morigi}},\ }\bibfield
  {title} {\bibinfo {title} {Reservoir-engineering shortcuts to adiabaticity},\
  }\href {https://doi.org/10.1103/PhysRevResearch.4.033005} {\bibfield
  {journal} {\bibinfo  {journal} {Phys. Rev. Research}\ }\textbf {\bibinfo
  {volume} {4}},\ \bibinfo {pages} {033005} (\bibinfo {year}
  {2022})}\BibitemShut {NoStop}%
\bibitem [{\citenamefont {Paviglianiti}\ \emph {et~al.}(2023)\citenamefont
  {Paviglianiti}, \citenamefont {Turkeshi}, \citenamefont {Schir{\`o}},\ and\
  \citenamefont {Silva}}]{paviglianiti2023}%
  \BibitemOpen
  \bibfield  {author} {\bibinfo {author} {\bibfnamefont {A.}~\bibnamefont
  {Paviglianiti}}, \bibinfo {author} {\bibfnamefont {X.}~\bibnamefont
  {Turkeshi}}, \bibinfo {author} {\bibfnamefont {M.}~\bibnamefont
  {Schir{\`o}}},\ and\ \bibinfo {author} {\bibfnamefont {A.}~\bibnamefont
  {Silva}},\ }\href {https://doi.org/10.48550/arXiv.2310.02686} {\bibinfo
  {title} {Enhanced {{Entanglement}} in the {{Measurement-Altered Quantum Ising
  Chain}}}} (\bibinfo {year} {2023}),\ \Eprint
  {https://arxiv.org/abs/2310.02686} {arxiv:2310.02686 [cond-mat,
  physics:quant-ph]} \BibitemShut {NoStop}%
\bibitem [{\citenamefont {Herasymenko}\ \emph {et~al.}(2023)\citenamefont
  {Herasymenko}, \citenamefont {Gornyi},\ and\ \citenamefont
  {Gefen}}]{HerasymenkoPRXQ2023}%
  \BibitemOpen
  \bibfield  {author} {\bibinfo {author} {\bibfnamefont {Y.}~\bibnamefont
  {Herasymenko}}, \bibinfo {author} {\bibfnamefont {I.}~\bibnamefont
  {Gornyi}},\ and\ \bibinfo {author} {\bibfnamefont {Y.}~\bibnamefont
  {Gefen}},\ }\bibfield  {title} {\bibinfo {title} {Measurement-driven
  navigation in many-body {H}ilbert space: {A}ctive-decision steering},\ }\href
  {https://doi.org/10.1103/PRXQuantum.4.020347} {\bibfield  {journal} {\bibinfo
   {journal} {PRX Quantum}\ }\textbf {\bibinfo {volume} {4}},\ \bibinfo {pages}
  {020347} (\bibinfo {year} {2023})}\BibitemShut {NoStop}%
\bibitem [{\citenamefont {Affleck}\ \emph {et~al.}(1987)\citenamefont
  {Affleck}, \citenamefont {Kennedy}, \citenamefont {Lieb},\ and\ \citenamefont
  {Tasaki}}]{AKLT1}%
  \BibitemOpen
  \bibfield  {author} {\bibinfo {author} {\bibfnamefont {I.}~\bibnamefont
  {Affleck}}, \bibinfo {author} {\bibfnamefont {T.}~\bibnamefont {Kennedy}},
  \bibinfo {author} {\bibfnamefont {E.~H.}\ \bibnamefont {Lieb}},\ and\
  \bibinfo {author} {\bibfnamefont {H.}~\bibnamefont {Tasaki}},\ }\bibfield
  {title} {\bibinfo {title} {Rigorous results on valence-bond ground states in
  antiferromagnets},\ }\href {https://doi.org/10.1103/PhysRevLett.59.799}
  {\bibfield  {journal} {\bibinfo  {journal} {Phys. Rev. Lett.}\ }\textbf
  {\bibinfo {volume} {59}},\ \bibinfo {pages} {799} (\bibinfo {year}
  {1987})}\BibitemShut {NoStop}%
\bibitem [{\citenamefont {Affleck}\ \emph {et~al.}(1988)\citenamefont
  {Affleck}, \citenamefont {Kennedy}, \citenamefont {Lieb},\ and\ \citenamefont
  {Tasaki}}]{AKLT2}%
  \BibitemOpen
  \bibfield  {author} {\bibinfo {author} {\bibfnamefont {I.}~\bibnamefont
  {Affleck}}, \bibinfo {author} {\bibfnamefont {T.}~\bibnamefont {Kennedy}},
  \bibinfo {author} {\bibfnamefont {E.~H.}\ \bibnamefont {Lieb}},\ and\
  \bibinfo {author} {\bibfnamefont {H.}~\bibnamefont {Tasaki}},\ }\bibfield
  {title} {\bibinfo {title} {Valence bond ground states in isotropic quantum
  antiferromagnets},\ }\href {https://doi.org/10.1007/BF01218021} {\bibfield
  {journal} {\bibinfo  {journal} {Commun.Math. Phys.}\ }\textbf {\bibinfo
  {volume} {115}},\ \bibinfo {pages} {477} (\bibinfo {year}
  {1988})}\BibitemShut {NoStop}%
\bibitem [{\citenamefont {Murta}\ \emph {et~al.}(2023)\citenamefont {Murta},
  \citenamefont {Cruz},\ and\ \citenamefont
  {Fern\'andez-Rossier}}]{MurtaPRR2023}%
  \BibitemOpen
  \bibfield  {author} {\bibinfo {author} {\bibfnamefont {B.}~\bibnamefont
  {Murta}}, \bibinfo {author} {\bibfnamefont {P.~M.~Q.}\ \bibnamefont {Cruz}},\
  and\ \bibinfo {author} {\bibfnamefont {J.}~\bibnamefont
  {Fern\'andez-Rossier}},\ }\bibfield  {title} {\bibinfo {title} {Preparing
  valence-bond-solid states on noisy intermediate-scale quantum computers},\
  }\href {https://doi.org/10.1103/PhysRevResearch.5.013190} {\bibfield
  {journal} {\bibinfo  {journal} {Phys. Rev. Res.}\ }\textbf {\bibinfo {volume}
  {5}},\ \bibinfo {pages} {013190} (\bibinfo {year} {2023})}\BibitemShut
  {NoStop}%
\bibitem [{\citenamefont {Smith}\ \emph {et~al.}(2023)\citenamefont {Smith},
  \citenamefont {Crane}, \citenamefont {Wiebe},\ and\ \citenamefont
  {Girvin}}]{SmithPRXQ2023}%
  \BibitemOpen
  \bibfield  {author} {\bibinfo {author} {\bibfnamefont {K.~C.}\ \bibnamefont
  {Smith}}, \bibinfo {author} {\bibfnamefont {E.}~\bibnamefont {Crane}},
  \bibinfo {author} {\bibfnamefont {N.}~\bibnamefont {Wiebe}},\ and\ \bibinfo
  {author} {\bibfnamefont {S.}~\bibnamefont {Girvin}},\ }\bibfield  {title}
  {\bibinfo {title} {{Deterministic Constant-Depth Preparation of the AKLT
  State on a Quantum Processor Using Fusion Measurements}},\ }\href
  {https://doi.org/10.1103/PRXQuantum.4.020315} {\bibfield  {journal} {\bibinfo
   {journal} {PRX Quantum}\ }\textbf {\bibinfo {volume} {4}},\ \bibinfo {pages}
  {020315} (\bibinfo {year} {2023})}\BibitemShut {NoStop}%
\bibitem [{\citenamefont {Chen}\ \emph {et~al.}(2023)\citenamefont {Chen},
  \citenamefont {Shen}, \citenamefont {Lee},\ and\ \citenamefont
  {Yang}}]{chen2023highfidelity}%
  \BibitemOpen
  \bibfield  {author} {\bibinfo {author} {\bibfnamefont {T.}~\bibnamefont
  {Chen}}, \bibinfo {author} {\bibfnamefont {R.}~\bibnamefont {Shen}}, \bibinfo
  {author} {\bibfnamefont {C.~H.}\ \bibnamefont {Lee}},\ and\ \bibinfo {author}
  {\bibfnamefont {B.}~\bibnamefont {Yang}},\ }\href
  {https://doi.org/10.48550/arXiv.2210.13840} {\bibinfo {title} {High-fidelity
  realization of the {AKLT} state on a {NISQ-era} quantum processor}} (\bibinfo
  {year} {2023}),\ \Eprint {https://arxiv.org/abs/2210.13840} {2210.13840}
  \BibitemShut {NoStop}%
\bibitem [{\citenamefont {Chen}\ and\ \citenamefont
  {Byrnes}(2023)}]{chen2023efficient}%
  \BibitemOpen
  \bibfield  {author} {\bibinfo {author} {\bibfnamefont {T.}~\bibnamefont
  {Chen}}\ and\ \bibinfo {author} {\bibfnamefont {T.}~\bibnamefont {Byrnes}},\
  }\href@noop {} {\bibinfo {title} {{Efficient preparation of the AKLT State
  with Measurement-based Imaginary Time Evolution}}} (\bibinfo {year} {2023}),\
  \Eprint {https://arxiv.org/abs/2310.06031} {arXiv:2310.06031 [quant-ph]}
  \BibitemShut {NoStop}%
\bibitem [{\citenamefont {Puente}\ \emph {et~al.}(2023)\citenamefont {Puente},
  \citenamefont {Motzoi}, \citenamefont {Calarco}, \citenamefont {Morigi},\
  and\ \citenamefont {Rizzi}}]{puente2023}%
  \BibitemOpen
  \bibfield  {author} {\bibinfo {author} {\bibfnamefont {D.~A.}\ \bibnamefont
  {Puente}}, \bibinfo {author} {\bibfnamefont {F.}~\bibnamefont {Motzoi}},
  \bibinfo {author} {\bibfnamefont {T.}~\bibnamefont {Calarco}}, \bibinfo
  {author} {\bibfnamefont {G.}~\bibnamefont {Morigi}},\ and\ \bibinfo {author}
  {\bibfnamefont {M.}~\bibnamefont {Rizzi}},\ }\href
  {https://doi.org/10.48550/arXiv.2305.08641} {\bibinfo {title} {Quantum state
  preparation via engineered ancilla resetting}} (\bibinfo {year} {2023}),\
  \Eprint {https://arxiv.org/abs/2309.00335} {arXiv:2309.00335} \BibitemShut
  {NoStop}%
\bibitem [{\citenamefont {Sriram}\ \emph {et~al.}(2023)\citenamefont {Sriram},
  \citenamefont {Rakovszky}, \citenamefont {Khemani},\ and\ \citenamefont
  {Ippoliti}}]{SriramPRB2023}%
  \BibitemOpen
  \bibfield  {author} {\bibinfo {author} {\bibfnamefont {A.}~\bibnamefont
  {Sriram}}, \bibinfo {author} {\bibfnamefont {T.}~\bibnamefont {Rakovszky}},
  \bibinfo {author} {\bibfnamefont {V.}~\bibnamefont {Khemani}},\ and\ \bibinfo
  {author} {\bibfnamefont {M.}~\bibnamefont {Ippoliti}},\ }\bibfield  {title}
  {\bibinfo {title} {Topology, criticality, and dynamically generated qubits in
  a stochastic measurement-only {K}itaev model},\ }\href
  {https://doi.org/10.1103/PhysRevB.108.094304} {\bibfield  {journal} {\bibinfo
   {journal} {Phys. Rev. B}\ }\textbf {\bibinfo {volume} {108}},\ \bibinfo
  {pages} {094304} (\bibinfo {year} {2023})}\BibitemShut {NoStop}%
\bibitem [{\citenamefont {Lavasani}\ \emph {et~al.}(2023)\citenamefont
  {Lavasani}, \citenamefont {Luo},\ and\ \citenamefont
  {Vijay}}]{LavasaniPRB2023}%
  \BibitemOpen
  \bibfield  {author} {\bibinfo {author} {\bibfnamefont {A.}~\bibnamefont
  {Lavasani}}, \bibinfo {author} {\bibfnamefont {Z.-X.}\ \bibnamefont {Luo}},\
  and\ \bibinfo {author} {\bibfnamefont {S.}~\bibnamefont {Vijay}},\ }\bibfield
   {title} {\bibinfo {title} {Monitored quantum dynamics and the {K}itaev spin
  liquid},\ }\href {https://doi.org/10.1103/PhysRevB.108.115135} {\bibfield
  {journal} {\bibinfo  {journal} {Phys. Rev. B}\ }\textbf {\bibinfo {volume}
  {108}},\ \bibinfo {pages} {115135} (\bibinfo {year} {2023})}\BibitemShut
  {NoStop}%
\bibitem [{\citenamefont {Morales}\ \emph {et~al.}(2023)\citenamefont
  {Morales}, \citenamefont {Gefen}, \citenamefont {Gornyi}, \citenamefont
  {Zazunov},\ and\ \citenamefont {Egger}}]{morales2023}%
  \BibitemOpen
  \bibfield  {author} {\bibinfo {author} {\bibfnamefont {S.}~\bibnamefont
  {Morales}}, \bibinfo {author} {\bibfnamefont {Y.}~\bibnamefont {Gefen}},
  \bibinfo {author} {\bibfnamefont {I.}~\bibnamefont {Gornyi}}, \bibinfo
  {author} {\bibfnamefont {A.}~\bibnamefont {Zazunov}},\ and\ \bibinfo {author}
  {\bibfnamefont {R.}~\bibnamefont {Egger}},\ }\href@noop {} {\bibinfo {title}
  {Engineering unsteerable quantum states with active feedback}} (\bibinfo
  {year} {2023}),\ \Eprint {https://arxiv.org/abs/2308.00384} {arXiv:2308.00384
  [quant-ph]} \BibitemShut {NoStop}%
\bibitem [{\citenamefont {Sayrin}\ \emph {et~al.}(2011)\citenamefont {Sayrin},
  \citenamefont {Dotsenko}, \citenamefont {Zhou}, \citenamefont {Peaudecerf},
  \citenamefont {Rybarczyk}, \citenamefont {Gleyzes}, \citenamefont {Rouchon},
  \citenamefont {Mirrahimi}, \citenamefont {Amini}, \citenamefont {Brune},
  \citenamefont {Raimond},\ and\ \citenamefont {Haroche}}]{sayrin2011}%
  \BibitemOpen
  \bibfield  {author} {\bibinfo {author} {\bibfnamefont {C.}~\bibnamefont
  {Sayrin}}, \bibinfo {author} {\bibfnamefont {I.}~\bibnamefont {Dotsenko}},
  \bibinfo {author} {\bibfnamefont {X.}~\bibnamefont {Zhou}}, \bibinfo {author}
  {\bibfnamefont {B.}~\bibnamefont {Peaudecerf}}, \bibinfo {author}
  {\bibfnamefont {T.}~\bibnamefont {Rybarczyk}}, \bibinfo {author}
  {\bibfnamefont {S.}~\bibnamefont {Gleyzes}}, \bibinfo {author} {\bibfnamefont
  {P.}~\bibnamefont {Rouchon}}, \bibinfo {author} {\bibfnamefont
  {M.}~\bibnamefont {Mirrahimi}}, \bibinfo {author} {\bibfnamefont
  {H.}~\bibnamefont {Amini}}, \bibinfo {author} {\bibfnamefont
  {M.}~\bibnamefont {Brune}}, \bibinfo {author} {\bibfnamefont {J.-M.}\
  \bibnamefont {Raimond}},\ and\ \bibinfo {author} {\bibfnamefont
  {S.}~\bibnamefont {Haroche}},\ }\bibfield  {title} {\bibinfo {title}
  {Real-time quantum feedback prepares and stabilizes photon number states},\
  }\href {https://doi.org/10.1038/nature10376} {\bibfield  {journal} {\bibinfo
  {journal} {Nature}\ }\textbf {\bibinfo {volume} {477}},\ \bibinfo {pages}
  {73} (\bibinfo {year} {2011})}\BibitemShut {NoStop}%
\bibitem [{\citenamefont {Majumdar}\ and\ \citenamefont
  {Ghosh}(2003)}]{MajumdarGhoshOriginal}%
  \BibitemOpen
  \bibfield  {author} {\bibinfo {author} {\bibfnamefont {C.~K.}\ \bibnamefont
  {Majumdar}}\ and\ \bibinfo {author} {\bibfnamefont {D.~K.}\ \bibnamefont
  {Ghosh}},\ }\bibfield  {title} {\bibinfo {title} {On
  next‐nearest‐neighbor interaction in linear chain. i},\ }\href
  {https://doi.org/10.1063/1.1664978} {\bibfield  {journal} {\bibinfo
  {journal} {Journal of Mathematical Physics}\ }\textbf {\bibinfo {volume}
  {10}},\ \bibinfo {pages} {1388} (\bibinfo {year} {2003})}\BibitemShut
  {NoStop}%
\bibitem [{Note1()}]{Note1}%
  \BibitemOpen
  \bibinfo {note} {In general, this may also include diagonal terms, leading to
  a Lamb shift in the system Hamiltonian $H_{\protect \text {s}}$.}\BibitemShut
  {Stop}%
\bibitem [{\citenamefont {Breuer}\ and\ \citenamefont
  {Petruccione}(2007)}]{Breuer_Petruccione07}%
  \BibitemOpen
  \bibfield  {author} {\bibinfo {author} {\bibfnamefont {H.-P.}\ \bibnamefont
  {Breuer}}\ and\ \bibinfo {author} {\bibfnamefont {F.}~\bibnamefont
  {Petruccione}},\ }\href
  {https://doi.org/10.1093/acprof:oso/9780199213900.001.0001} {\emph {\bibinfo
  {title} {The Theory of Open Quantum Systems}}}\ (\bibinfo  {publisher}
  {Oxford University Press, Oxford},\ \bibinfo {year} {2007})\BibitemShut
  {NoStop}%
\bibitem [{\citenamefont {Rivas}\ and\ \citenamefont
  {Huelga}(2011)}]{Rivas_Huelga}%
  \BibitemOpen
  \bibfield  {author} {\bibinfo {author} {\bibfnamefont {A.}~\bibnamefont
  {Rivas}}\ and\ \bibinfo {author} {\bibfnamefont {S.~F.}\ \bibnamefont
  {Huelga}},\ }\href {https://doi.org/10.1007/978-3-642-23354-8} {\emph
  {\bibinfo {title} {Open Quantum Systems}}}\ (\bibinfo  {publisher} {Springer
  Berlin, Heidelberg},\ \bibinfo {year} {2011})\BibitemShut {NoStop}%
\bibitem [{Note2()}]{Note2}%
  \BibitemOpen
  \bibinfo {note} {We assume here that $\protect \mathcal L$ has no purely
  imaginary eigenvalues~\cite {yoshida2023SteadyStateUniqueness}}\BibitemShut
  {NoStop}%
\bibitem [{\citenamefont {Basilewitsch}\ \emph {et~al.}(2019)\citenamefont
  {Basilewitsch}, \citenamefont {Koch},\ and\ \citenamefont
  {Reich}}]{BasilewitschAQT19}%
  \BibitemOpen
  \bibfield  {author} {\bibinfo {author} {\bibfnamefont {D.}~\bibnamefont
  {Basilewitsch}}, \bibinfo {author} {\bibfnamefont {C.~P.}\ \bibnamefont
  {Koch}},\ and\ \bibinfo {author} {\bibfnamefont {D.~M.}\ \bibnamefont
  {Reich}},\ }\bibfield  {title} {\bibinfo {title} {Quantum optimal control for
  mixed state squeezing in cavity optomechanics},\ }\href
  {https://doi.org/10.1002/qute.201800110} {\bibfield  {journal} {\bibinfo
  {journal} {Adv. Quantum Technol.}\ }\textbf {\bibinfo {volume} {2}},\
  \bibinfo {pages} {1800110} (\bibinfo {year} {2019})}\BibitemShut {NoStop}%
\bibitem [{\citenamefont {Yoshida}(2023)}]{yoshida2023SteadyStateUniqueness}%
  \BibitemOpen
  \bibfield  {author} {\bibinfo {author} {\bibfnamefont {H.}~\bibnamefont
  {Yoshida}},\ }\href@noop {} {\bibinfo {title} {{Uniqueness of steady states
  of Gorini-Kossakowski-Sudarshan-Lindblad equations: a simple proof}}}
  (\bibinfo {year} {2023}),\ \Eprint {https://arxiv.org/abs/2309.00335}
  {arXiv:2309.00335} \BibitemShut {NoStop}%
\bibitem [{\citenamefont {Pollmann}\ \emph {et~al.}(2012)\citenamefont
  {Pollmann}, \citenamefont {Berg}, \citenamefont {Turner},\ and\ \citenamefont
  {Oshikawa}}]{PollmannPRB2012}%
  \BibitemOpen
  \bibfield  {author} {\bibinfo {author} {\bibfnamefont {F.}~\bibnamefont
  {Pollmann}}, \bibinfo {author} {\bibfnamefont {E.}~\bibnamefont {Berg}},
  \bibinfo {author} {\bibfnamefont {A.~M.}\ \bibnamefont {Turner}},\ and\
  \bibinfo {author} {\bibfnamefont {M.}~\bibnamefont {Oshikawa}},\ }\bibfield
  {title} {\bibinfo {title} {Symmetry protection of topological phases in
  one-dimensional quantum spin systems},\ }\href
  {https://doi.org/10.1103/PhysRevB.85.075125} {\bibfield  {journal} {\bibinfo
  {journal} {Phys. Rev. B}\ }\textbf {\bibinfo {volume} {85}},\ \bibinfo
  {pages} {075125} (\bibinfo {year} {2012})}\BibitemShut {NoStop}%
\bibitem [{\citenamefont {Verstraete}\ and\ \citenamefont
  {Cirac}(2004)}]{VerstraetePRA2004}%
  \BibitemOpen
  \bibfield  {author} {\bibinfo {author} {\bibfnamefont {F.}~\bibnamefont
  {Verstraete}}\ and\ \bibinfo {author} {\bibfnamefont {J.~I.}\ \bibnamefont
  {Cirac}},\ }\bibfield  {title} {\bibinfo {title} {Valence-bond states for
  quantum computation},\ }\href {https://doi.org/10.1103/PhysRevA.70.060302}
  {\bibfield  {journal} {\bibinfo  {journal} {Phys. Rev. A}\ }\textbf {\bibinfo
  {volume} {70}},\ \bibinfo {pages} {060302} (\bibinfo {year}
  {2004})}\BibitemShut {NoStop}%
\bibitem [{\citenamefont {Wang}\ \emph {et~al.}(2023)\citenamefont {Wang},
  \citenamefont {Snizhko}, \citenamefont {Romito}, \citenamefont {Gefen},\ and\
  \citenamefont {Murch}}]{PhysRevA.108.013712}%
  \BibitemOpen
  \bibfield  {author} {\bibinfo {author} {\bibfnamefont {Y.}~\bibnamefont
  {Wang}}, \bibinfo {author} {\bibfnamefont {K.}~\bibnamefont {Snizhko}},
  \bibinfo {author} {\bibfnamefont {A.}~\bibnamefont {Romito}}, \bibinfo
  {author} {\bibfnamefont {Y.}~\bibnamefont {Gefen}},\ and\ \bibinfo {author}
  {\bibfnamefont {K.}~\bibnamefont {Murch}},\ }\bibfield  {title} {\bibinfo
  {title} {Dissipative preparation and stabilization of many-body quantum
  states in a superconducting qutrit array},\ }\href
  {https://doi.org/10.1103/PhysRevA.108.013712} {\bibfield  {journal} {\bibinfo
   {journal} {Phys. Rev. A}\ }\textbf {\bibinfo {volume} {108}},\ \bibinfo
  {pages} {013712} (\bibinfo {year} {2023})}\BibitemShut {NoStop}%
\bibitem [{\citenamefont {Ticozzi}\ and\ \citenamefont
  {Viola}(2012)}]{ticozzi_stabilizing_2012}%
  \BibitemOpen
  \bibfield  {author} {\bibinfo {author} {\bibfnamefont {F.}~\bibnamefont
  {Ticozzi}}\ and\ \bibinfo {author} {\bibfnamefont {L.}~\bibnamefont
  {Viola}},\ }\bibfield  {title} {\bibinfo {title} {Stabilizing entangled
  states with quasi-local quantum dynamical semigroups},\ }\href
  {https://doi.org/10.1098/rsta.2011.0485} {\bibfield  {journal} {\bibinfo
  {journal} {Philosophical Transactions of the Royal Society A: Mathematical,
  Physical and Engineering Sciences}\ }\textbf {\bibinfo {volume} {370}},\
  \bibinfo {pages} {5259} (\bibinfo {year} {2012})}\BibitemShut {NoStop}%
\bibitem [{\citenamefont {Dalibard}\ \emph {et~al.}(1992)\citenamefont
  {Dalibard}, \citenamefont {Castin},\ and\ \citenamefont
  {M\o{}lmer}}]{montecarlo1}%
  \BibitemOpen
  \bibfield  {author} {\bibinfo {author} {\bibfnamefont {J.}~\bibnamefont
  {Dalibard}}, \bibinfo {author} {\bibfnamefont {Y.}~\bibnamefont {Castin}},\
  and\ \bibinfo {author} {\bibfnamefont {K.}~\bibnamefont {M\o{}lmer}},\
  }\bibfield  {title} {\bibinfo {title} {Wave-function approach to dissipative
  processes in quantum optics},\ }\href
  {https://doi.org/10.1103/PhysRevLett.68.580} {\bibfield  {journal} {\bibinfo
  {journal} {Phys. Rev. Lett.}\ }\textbf {\bibinfo {volume} {68}},\ \bibinfo
  {pages} {580} (\bibinfo {year} {1992})}\BibitemShut {NoStop}%
\bibitem [{Note3()}]{Note3}%
  \BibitemOpen
  \bibinfo {note} {The bootstrap resampling method \cite {bootsrapResampling}
  is useful to estimate quantities from a random distribution with low sample
  size. It also provides a simple way to assign confidence intervals and
  therefore standard deviations. The method works by repeatedly (here: 500
  times) randomly resampling the data with replacement. Each of the resamples
  is used to estimate the quantities of interest and the final mean and
  variance are given by the resulting statistics.}\BibitemShut {Stop}%
\bibitem [{Note4()}]{Note4}%
  \BibitemOpen
  \bibinfo {note} {To be precise, non-dilute cooling towards a GHZ state
  requires conditional stabilization, i.e., active steering~\cite
  {ticozzi_steady-state_2014}. Passive steering is possible when targeting a
  GHZ state only approximately~\cite {martin2023}}\BibitemShut {NoStop}%
\bibitem [{\citenamefont {Morigi}\ \emph {et~al.}(2015)\citenamefont {Morigi},
  \citenamefont {Eschner}, \citenamefont {Cormick}, \citenamefont {Lin},
  \citenamefont {Leibfried},\ and\ \citenamefont {Wineland}}]{Morigi:2015}%
  \BibitemOpen
  \bibfield  {author} {\bibinfo {author} {\bibfnamefont {G.}~\bibnamefont
  {Morigi}}, \bibinfo {author} {\bibfnamefont {J.}~\bibnamefont {Eschner}},
  \bibinfo {author} {\bibfnamefont {C.}~\bibnamefont {Cormick}}, \bibinfo
  {author} {\bibfnamefont {Y.}~\bibnamefont {Lin}}, \bibinfo {author}
  {\bibfnamefont {D.}~\bibnamefont {Leibfried}},\ and\ \bibinfo {author}
  {\bibfnamefont {D.~J.}\ \bibnamefont {Wineland}},\ }\bibfield  {title}
  {\bibinfo {title} {Dissipative quantum control of a spin chain},\ }\href
  {https://doi.org/10.1103/PhysRevLett.115.200502} {\bibfield  {journal}
  {\bibinfo  {journal} {Phys. Rev. Lett.}\ }\textbf {\bibinfo {volume} {115}},\
  \bibinfo {pages} {200502} (\bibinfo {year} {2015})}\BibitemShut {NoStop}%
\bibitem [{\citenamefont {D'Alessandro}(2007)}]{dalessandro_introduction_2007}%
  \BibitemOpen
  \bibfield  {author} {\bibinfo {author} {\bibfnamefont {D.}~\bibnamefont
  {D'Alessandro}},\ }\href {https://doi.org/10.1201/9781584888833} {\emph
  {\bibinfo {title} {Introduction to Quantum Control and Dynamics}}},\ \bibinfo
  {edition} {0th}\ ed.\ (\bibinfo  {publisher} {Chapman and Hall/{CRC}},\
  \bibinfo {year} {2007})\BibitemShut {NoStop}%
\bibitem [{\citenamefont {Roy}\ \emph {et~al.}(2023)\citenamefont {Roy},
  \citenamefont {Otto}, \citenamefont {Menu},\ and\ \citenamefont
  {Morigi}}]{Roy:2023}%
  \BibitemOpen
  \bibfield  {author} {\bibinfo {author} {\bibfnamefont {S.}~\bibnamefont
  {Roy}}, \bibinfo {author} {\bibfnamefont {C.}~\bibnamefont {Otto}}, \bibinfo
  {author} {\bibfnamefont {R.}~\bibnamefont {Menu}},\ and\ \bibinfo {author}
  {\bibfnamefont {G.}~\bibnamefont {Morigi}},\ }\bibfield  {title} {\bibinfo
  {title} {Rise and fall of entanglement between two qubits in a non-markovian
  bath},\ }\href {https://doi.org/10.1103/PhysRevA.108.032205} {\bibfield
  {journal} {\bibinfo  {journal} {Phys. Rev. A}\ }\textbf {\bibinfo {volume}
  {108}},\ \bibinfo {pages} {032205} (\bibinfo {year} {2023})}\BibitemShut
  {NoStop}%
\bibitem [{\citenamefont {Moudgalya}\ \emph {et~al.}(2018)\citenamefont
  {Moudgalya}, \citenamefont {Rachel}, \citenamefont {Bernevig},\ and\
  \citenamefont {Regnault}}]{Moudgalya2018}%
  \BibitemOpen
  \bibfield  {author} {\bibinfo {author} {\bibfnamefont {S.}~\bibnamefont
  {Moudgalya}}, \bibinfo {author} {\bibfnamefont {S.}~\bibnamefont {Rachel}},
  \bibinfo {author} {\bibfnamefont {B.~A.}\ \bibnamefont {Bernevig}},\ and\
  \bibinfo {author} {\bibfnamefont {N.}~\bibnamefont {Regnault}},\ }\bibfield
  {title} {\bibinfo {title} {Exact excited states of nonintegrable models},\
  }\href {https://doi.org/10.1103/PhysRevB.98.235155} {\bibfield  {journal}
  {\bibinfo  {journal} {Phys. Rev. B}\ }\textbf {\bibinfo {volume} {98}},\
  \bibinfo {pages} {235155} (\bibinfo {year} {2018})}\BibitemShut {NoStop}%
\bibitem [{Sup()}]{SuppMat}%
  \BibitemOpen
  \href@noop {} {\bibinfo {title} {See supplemental information for technical
  details.}}\BibitemShut {Stop}%
\bibitem [{\citenamefont {Defenu}\ \emph {et~al.}(2023)\citenamefont {Defenu},
  \citenamefont {Donner}, \citenamefont {Macr\`{\i}}, \citenamefont {Pagano},
  \citenamefont {Ruffo},\ and\ \citenamefont {Trombettoni}}]{defenu:2023}%
  \BibitemOpen
  \bibfield  {author} {\bibinfo {author} {\bibfnamefont {N.}~\bibnamefont
  {Defenu}}, \bibinfo {author} {\bibfnamefont {T.}~\bibnamefont {Donner}},
  \bibinfo {author} {\bibfnamefont {T.}~\bibnamefont {Macr\`{\i}}}, \bibinfo
  {author} {\bibfnamefont {G.}~\bibnamefont {Pagano}}, \bibinfo {author}
  {\bibfnamefont {S.}~\bibnamefont {Ruffo}},\ and\ \bibinfo {author}
  {\bibfnamefont {A.}~\bibnamefont {Trombettoni}},\ }\bibfield  {title}
  {\bibinfo {title} {Long-range interacting quantum systems},\ }\href
  {https://doi.org/10.1103/RevModPhys.95.035002} {\bibfield  {journal}
  {\bibinfo  {journal} {Rev. Mod. Phys.}\ }\textbf {\bibinfo {volume} {95}},\
  \bibinfo {pages} {035002} (\bibinfo {year} {2023})}\BibitemShut {NoStop}%
\bibitem [{\citenamefont {Tran}\ \emph {et~al.}(2020)\citenamefont {Tran},
  \citenamefont {Chen}, \citenamefont {Ehrenberg}, \citenamefont {Guo},
  \citenamefont {Deshpande}, \citenamefont {Hong}, \citenamefont {Gong},
  \citenamefont {Gorshkov},\ and\ \citenamefont {Lucas}}]{Tran:2020}%
  \BibitemOpen
  \bibfield  {author} {\bibinfo {author} {\bibfnamefont {M.~C.}\ \bibnamefont
  {Tran}}, \bibinfo {author} {\bibfnamefont {C.-F.}\ \bibnamefont {Chen}},
  \bibinfo {author} {\bibfnamefont {A.}~\bibnamefont {Ehrenberg}}, \bibinfo
  {author} {\bibfnamefont {A.~Y.}\ \bibnamefont {Guo}}, \bibinfo {author}
  {\bibfnamefont {A.}~\bibnamefont {Deshpande}}, \bibinfo {author}
  {\bibfnamefont {Y.}~\bibnamefont {Hong}}, \bibinfo {author} {\bibfnamefont
  {Z.-X.}\ \bibnamefont {Gong}}, \bibinfo {author} {\bibfnamefont {A.~V.}\
  \bibnamefont {Gorshkov}},\ and\ \bibinfo {author} {\bibfnamefont
  {A.}~\bibnamefont {Lucas}},\ }\bibfield  {title} {\bibinfo {title} {Hierarchy
  of linear light cones with long-range interactions},\ }\href
  {https://doi.org/10.1103/PhysRevX.10.031009} {\bibfield  {journal} {\bibinfo
  {journal} {Phys. Rev. X}\ }\textbf {\bibinfo {volume} {10}},\ \bibinfo
  {pages} {031009} (\bibinfo {year} {2020})}\BibitemShut {NoStop}%
\bibitem [{\citenamefont {Zhang}\ and\ \citenamefont
  {Barthel}(2023)}]{zhang2023}%
  \BibitemOpen
  \bibfield  {author} {\bibinfo {author} {\bibfnamefont {Y.}~\bibnamefont
  {Zhang}}\ and\ \bibinfo {author} {\bibfnamefont {T.}~\bibnamefont
  {Barthel}},\ }\href {https://doi.org/10.48550/arXiv.2310.17641} {\bibinfo
  {title} {Criteria for {{Davies Irreducibility}} of {{Markovian Quantum
  Dynamics}}}} (\bibinfo {year} {2023}),\ \Eprint
  {https://arxiv.org/abs/2310.17641} {arxiv:2310.17641 [quant-ph]} \BibitemShut
  {NoStop}%
\bibitem [{\citenamefont {Prosen}(2008)}]{prosen_third_2008}%
  \BibitemOpen
  \bibfield  {author} {\bibinfo {author} {\bibfnamefont {T.}~\bibnamefont
  {Prosen}},\ }\bibfield  {title} {\bibinfo {title} {Third quantization: a
  general method to solve master equations for quadratic open {F}ermi
  systems},\ }\href {https://doi.org/10.1088/1367-2630/10/4/043026} {\bibfield
  {journal} {\bibinfo  {journal} {New Journal of Physics}\ }\textbf {\bibinfo
  {volume} {10}},\ \bibinfo {pages} {043026} (\bibinfo {year} {2008})},\
  \Eprint {https://arxiv.org/abs/0801.1257} {0801.1257} \BibitemShut {NoStop}%
\bibitem [{\citenamefont {Burgarth}\ and\ \citenamefont
  {Giovannetti}(2007{\natexlab{c}})}]{burgarth_mediated_2007}%
  \BibitemOpen
  \bibfield  {author} {\bibinfo {author} {\bibfnamefont {D.}~\bibnamefont
  {Burgarth}}\ and\ \bibinfo {author} {\bibfnamefont {V.}~\bibnamefont
  {Giovannetti}},\ }\bibfield  {title} {\bibinfo {title} {Mediated
  homogenization},\ }\href {https://doi.org/10.1103/PhysRevA.76.062307}
  {\bibfield  {journal} {\bibinfo  {journal} {Physical Review A}\ }\textbf
  {\bibinfo {volume} {76}},\ \bibinfo {pages} {062307} (\bibinfo {year}
  {2007}{\natexlab{c}})}\BibitemShut {NoStop}%
\bibitem [{\citenamefont {Efron}(1979)}]{bootsrapResampling}%
  \BibitemOpen
  \bibfield  {author} {\bibinfo {author} {\bibfnamefont {B.}~\bibnamefont
  {Efron}},\ }\bibfield  {title} {\bibinfo {title} {{Bootstrap Methods: Another
  Look at the Jackknife}},\ }\href {https://doi.org/10.1214/aos/1176344552}
  {\bibfield  {journal} {\bibinfo  {journal} {The Annals of Statistics}\
  }\textbf {\bibinfo {volume} {7}},\ \bibinfo {pages} {1 } (\bibinfo {year}
  {1979})}\BibitemShut {NoStop}%
\bibitem [{\citenamefont {Fogarty}\ \emph {et~al.}(2016)\citenamefont
  {Fogarty}, \citenamefont {Landa}, \citenamefont {Cormick},\ and\
  \citenamefont {Morigi}}]{Fogarty:2016}%
  \BibitemOpen
  \bibfield  {author} {\bibinfo {author} {\bibfnamefont {T.}~\bibnamefont
  {Fogarty}}, \bibinfo {author} {\bibfnamefont {H.}~\bibnamefont {Landa}},
  \bibinfo {author} {\bibfnamefont {C.}~\bibnamefont {Cormick}},\ and\ \bibinfo
  {author} {\bibfnamefont {G.}~\bibnamefont {Morigi}},\ }\bibfield  {title}
  {\bibinfo {title} {Optomechanical many-body cooling to the ground state using
  frustration},\ }\href {https://doi.org/10.1103/PhysRevA.94.023844} {\bibfield
   {journal} {\bibinfo  {journal} {Phys. Rev. A}\ }\textbf {\bibinfo {volume}
  {94}},\ \bibinfo {pages} {023844} (\bibinfo {year} {2016})}\BibitemShut
  {NoStop}%
\bibitem [{\citenamefont {Englert}\ and\ \citenamefont
  {Morigi}(2002)}]{Englert}%
  \BibitemOpen
  \bibfield  {author} {\bibinfo {author} {\bibfnamefont {B.-G.}\ \bibnamefont
  {Englert}}\ and\ \bibinfo {author} {\bibfnamefont {G.}~\bibnamefont
  {Morigi}},\ }\href {https://doi.org/0.1007/3-540-45855-7_2} {\emph {\bibinfo
  {title} {Five Lectures on Dissipative Master Equations. In: Buchleitner, A.,
  Hornberger, K. (eds) Coherent Evolution in Noisy Environments. Lecture Notes
  in Physics}}}\ (\bibinfo  {publisher} {Springer},\ \bibinfo {year}
  {2002})\BibitemShut {NoStop}%
\bibitem [{\citenamefont {Martin}\ and\ \citenamefont
  {Sarlette}(2023)}]{martin2023}%
  \BibitemOpen
  \bibfield  {author} {\bibinfo {author} {\bibfnamefont {V.}~\bibnamefont
  {Martin}}\ and\ \bibinfo {author} {\bibfnamefont {A.}~\bibnamefont
  {Sarlette}},\ }\href {https://doi.org/10.48550/arXiv.2306.05070} {\bibinfo
  {title} {Stabilization of approximate {{GHZ}} state with quasi-local
  couplings}} (\bibinfo {year} {2023}),\ \Eprint
  {https://arxiv.org/abs/2306.05070} {arxiv:2306.05070 [quant-ph]} \BibitemShut
  {NoStop}%
\end{thebibliography}%

\end{document}